\documentclass[11pt]{article}
\usepackage{microtype}
\usepackage[utf8]{inputenc}
\usepackage[margin=1in]{geometry}
\usepackage{amsthm,amsmath,amsfonts,amssymb}
\usepackage{thmtools}
\usepackage{thm-restate}
\usepackage{mathtools}
\usepackage[dvipsnames]{xcolor}
\usepackage{tcolorbox}
\usepackage{enumitem}
\setlist[itemize]{label=\text{\tiny$\blacksquare$}}
\usepackage{xfrac}
\usepackage{standalone}
\usepackage{tikz}
\usepackage[
    backend=biber,
    style=alphabetic,
    maxalphanames=4,
    minalphanames=3,
    maxbibnames=99,
    giveninits=false,
    doi=true,            
    isbn=false,           
    url=false,
    useprefix=false,
]{biblatex}

\DeclareSortingTemplate{anyt}{ 
  \sort{
    \field{presort}
  }
  \sort{
    \field{sortkey}
  }
  \sort{
    \field{sortname}
    \field{author}
    \field{editor}
    \field{translator}
    \field{sorttitle}
    \field{title}
  }
  \sort{
    \field{sortyear}
    \field{year}
  }
}

\usepackage{xurl}
\usepackage[colorlinks, linkcolor=OliveGreen, citecolor=purple,urlcolor=purple, bookmarks=true]{hyperref}
\usepackage{algorithm} 
\usepackage[noend]{algpseudocode} 
\usepackage[capitalise,nameinlink]{cleveref}

\DeclareFieldFormat{eprint:arxiv}{%
  \ifhyperref
    {\href{https://arxiv.org/abs/#1}{arXiv\addcolon\nolinkurl{#1}}}
    {arXiv\addcolon\nolinkurl{#1}}}

\DeclareFieldFormat{eprint:eccc}{%
\href{https://eccc.weizmann.ac.il/report/#1/}{ECCC:\space#1}%
}

\DeclareFieldFormat{doi}{%
  \ifhyperref
    {\href{https://doi.org/#1}{doi\addcolon\nolinkurl{#1}}}
    {doi\addcolon\nolinkurl{#1}}}
\DeclareFieldAlias{eprint}{eprint:arxiv}

\setenumerate[1]{wide = 0pt,labelwidth = 0em, itemsep = 0em, leftmargin=*}
\newlist{defaultenumerate}{enumerate}{3}
\setlist[defaultenumerate,1]{label=\arabic*.}%
\setlist[defaultenumerate,2]{label=\arabic*.}%
\setlist[defaultenumerate,3]{label=\arabic*.}%

\usetikzlibrary{arrows.meta}
\usetikzlibrary{decorations.markings}

\makeatletter
\def\@cite#1#2{\textup{[{#1\if@tempswa , #2\fi}]}}
\makeatother

\theoremstyle{plain}
\newtheorem{theorem}{Theorem}
\newtheorem{corollary}{Corollary}
\newtheorem{lemma}{Lemma}
\newtheorem{fact}{Fact}
\newtheorem{conjecture}{Conjecture}

\theoremstyle{definition}
\newtheorem{definition}{Definition}
\newtheorem{remark}{Remark}

\newcommand{\ket}[1]{| #1 \rangle}
\newcommand{\abs}[1]{\left| #1 \right|}
\newcommand{\calA}{\mathcal{A}}
\newcommand{\calB}{\mathcal{B}}
\newcommand{\calM}{\mathcal{M}}
\newcommand{\lift}{\mathrm{lift}}
\newcommand{\lifts}{\mathrm{lifts}}
\newcommand{\boldzero}{\mathbf{0}}
\DeclareMathOperator{\PostQ}{PostQ}
\DeclareMathOperator{\NQ}{NQ}
\DeclareMathOperator{\rdeg}{rdeg}
\DeclareMathOperator{\ndeg}{ndeg}
\DeclareMathOperator{\sdeg}{deg_{\pm}}

\DeclareMathOperator{\odeg}{odeg}
\DeclareMathOperator{\expect}{E}
\DeclareMathOperator{\ber}{B}
\DeclareMathOperator{\bs}{bs}
\DeclareMathOperator{\fbs}{fbs}
\DeclareMathOperator{\supp}{supp}
\DeclareMathOperator{\rc}{RC}
\DeclareMathOperator{\hit}{hit}
\DeclareMathOperator{\unhit}{unhit}
\DeclareMathOperator{\argmin}{argmin}
\DeclareMathOperator{\PAND}{P_{\mathrm{AND}}}
\DeclareMathOperator{\Inf}{Inf}
\DeclareMathOperator{\AS}{AS}
\DeclareMathOperator{\AND}{AND}
\DeclareMathOperator{\OR}{OR}
\DeclareMathOperator{\NOR}{NOR}
\DeclareMathOperator{\ANDOR}{AND-OR}
\DeclareMathOperator{\PARITY}{PARITY}

\DeclareMathOperator{\Cmin}{C_{min}}
\DeclareMathOperator{\Cmindown}{C_{min}^{\downarrow}}
\DeclareMathOperator{\bsmin}{bs_{min}}
\DeclareMathOperator{\C}{C}
\DeclareMathOperator{\D}{D}
\DeclareMathOperator{\R}{R}
\DeclareMathOperator{\s}{s}
\DeclareMathOperator{\fbsmin}{fbs_{min}}
\DeclareMathOperator{\fbsmindown}{fbs_{min}^{\downarrow}}
\DeclareMathOperator{\RC}{RC}

\DeclareMathOperator{\RCmindown}{RC_{min}^{\downarrow}}
\newcommand{\tO}{\widetilde{O}}

\algnewcommand{\Let}{let }
\algnewcommand\algorithmicinput{\textbf{Input:}}
\algnewcommand\algorithmicoutput{\textbf{Output:}}

\renewcommand{\emptyset}{\varnothing}

\title{\vspace{-6ex} \bfseries Rational degree is polynomially related to degree}

\author{\normalsize
Robin Kothari\footnotemark[1]\qquad
Matt Kovacs-Deak\footnotemark[2]\qquad
Daochen Wang\footnotemark[3]\qquad
Rain Zimin Yang\footnotemark[4]
}

\date{\vspace{-5ex}}

\begin{document}

\maketitle

\begin{abstract}\noindent
We prove that $\deg(f) \leq \tO(\rdeg(f)^3)$ for every Boolean function $f$, where $\deg(f)$ is the degree of $f$ and $\rdeg(f)$ is the rational degree of $f$. This resolves the second of the three open problems stated by Nisan and Szegedy, and attributed to Fortnow, in 1994 \cite{degree_nisan_1994}.
\end{abstract}

\makeatletter

\begingroup
\renewcommand{\thefootnote}{\fnsymbol{footnote}}
\renewcommand{\@makefntext}[1]{#1}
\footnotetext[1]{%
\footnotemark[1]Google\;\;\;
\footnotemark[2]University of Maryland \;\;\;
\footnotemark[3]University of British Columbia \;\;\;
\footnotemark[4]University of British Columbia
}

\footnotetext[3]{
\footnotemark[3]Corresponding author: \texttt{wdaochen@gmail.com}}
\endgroup

\makeatother

\vspace{-2ex}

\section{Introduction}

The degree $\deg(f)$ of a Boolean function $f\colon \{0,1\}^n \to \{0,1\}$ is the minimum value of $\deg(r)$ such that $r$ is a real polynomial and $f = r$ on $\{0,1\}^n$. The rational degree $\rdeg(f)$ of $f$ is the minimum value of $\max(\deg(p),\deg(q))$ such that $p,q$ are real polynomials and $f=p/q$ on $\{0,1\}^n$. Clearly, $\rdeg(f) \leq \deg(f)$ as can be seen by taking the denominator $q$ to be $1$. On the other hand, it is unclear whether $\rdeg(f)$ could be much smaller than $\deg(f)$.

Degree is polynomially related to almost all Boolean complexity measures. Rational degree characterizes the exact postselected quantum query complexity \cite{postqe_mahadev_2015}. Whether rational degree is polynomially related to degree has remained an open problem since it was first stated by Nisan and Szegedy, and attributed to Fortnow, over three decades ago \cite{degree_nisan_1994}.

Since then, the problem has been reiterated in works such as \cite{ndeg_dewolf_2000,thesis_dewolf_2001,nondeterministic_dewolf_2003,video_sherstov_2013,postqe_mahadev_2015,thesis_cade_2020,implications_aaronson_2021,rdeg_iyer_2025}. Rational degree can also be motivated from a variety of perspectives beyond quantum postselection:
\begin{enumerate}
    \item The analogous randomized query measure, exact postselected randomized query complexity, equals the \emph{certificate complexity} \cite[Theorem 16]{thesis_cade_2020}. Therefore, rational degree can be viewed as a natural quantum notion of certificate complexity. (Note that rational degree is distinct from Aaronson's definition of quantum certificate complexity \cite{qcertificate_aaronson_2008}.)

    \item We have $\sdeg(f) \leq 2\rdeg(f)$ by considering $\epsilon-p^2$ for small $\epsilon>0$, where $\sdeg(f)$ is the \emph{sign degree} (or \emph{threshold degree}) of $f$, that is, the minimum degree of a real polynomial that agrees in sign with $(-1)^{f(x)}$ for all $x\in \{0,1\}^n$. Therefore, $\sdeg(f)/2 \leq \rdeg(f) \leq \deg(f)$.
    Sign degree and degree are well-studied complexity measures \cite{slicing_saks_1993,degree_nisan_1994,analysis_odonnell_2014}. As rational degree inherits structural properties of both, it helps us better understand their relationship.
    
    \item It is not hard to see that $\rdeg(f) = \max(\ndeg(f),\ndeg(\neg f))$, where $\neg$ denotes negation and $\ndeg(f)$ is the minimum degree of a real polynomial $s$ such that, for all $x \in \{0,1\}^n$, $s(x) = 0$ if and only if $f(x) = 0$. (This is \cref{fact:rdeg_ndeg}, for which we give a proof for completeness.) In combinatorics, the Alon-Füredi theorem \cite{covering_alon_1993} is equivalent to $\ndeg(\AND_n) = n$, and the key lemma of \cite{balancing_alon_1998} essentially shows $\ndeg(\PARITY_n) \geq n/2$. At a deeper level, the relation between rational degree and degree governs the \emph{effectiveness} of a natural Nullstellensatz for the hypercube, in the sense of \cite{bounds_brownawell_1987,sharp_kollar_1988,nullstellensatz_alon_1999,effective_jelonek_2005}. We give more details in \cref{sec:nullstellensatz}.
    
    \item In complexity theory, $\ndeg(f)$ equals the \emph{degree of the set} $f^{-1}(0)$ as defined by Smolensky \cite{lowdegree_smolensky_1993}, and equals the \emph{one-sided $0$-approximate degree} of $f$, up to a factor of $2$, as defined by Sherstov \cite{breaking_sherstov_2018}. Polynomially relating rational degree to degree also implies  
    $\mathbf{P} = (\text{C}_{=}\mathbf{P} \cap \text{co-C}_{=}\mathbf{P})$ with respect to generic oracles \cite{oracle_fortnow_2003}. This originally motivated Fortnow to pose his open problem \cite{email_fortnow_2026},  which he has described as one of his ``favorite and most frustrating'' \cite{blog_fortnow_2003}. 
\end{enumerate}

In this work, we resolve the open problem by proving $\deg(f) \leq \tO(\rdeg(f)^3)$ for every Boolean function $f$. In fact, we prove $\deg(f) \leq \tO(\sdeg(f)^2\rdeg(f))$, which is stronger.

Before presenting our strongest results in \cref{sec:improved_lower_bound}, we establish $\deg(f) \leq 16 \rdeg(f)^4$ 
in \cref{sec:poly_lower_bound}. While weaker, this result already resolves Fortnow's open problem and serves as a gentle introduction to the techniques used later. 
We end by discussing a variety of implications and open problems in \cref{sec:implications}.

\section{Preliminaries}
\label{sec:prelims}

This section presents the main definitions and tools used in our proof. For more context, we refer the reader to de Wolf's thesis \cite{thesis_dewolf_2001} and the survey of Buhrman and de Wolf \cite{survey_buhrman_2002}.

For a positive integer $n$, we write $[n]$ for the set $\{1,\dots, n\}$. We use the symbol $\bigsqcup$ for disjoint union of sets. We employ standard notation for polynomial rings and their fraction fields. For $p\in \mathbb{R}[X_1,\dots,X_n]$, we write $\deg(p)$ for the degree of $p$. If $p$ is of the form $\sum_{S\subseteq [n]} c_S\cdot \prod_{i\in S}X_i$, then we say it is multilinear. In this case, the degree of $p$ is equal to $\max\{\abs{S} \colon c_S \neq 0\}$. For $x\in \{0,1\}^n$, we write $\abs{x}$ for the Hamming weight of $x$, that is, the number of $1$s in $x$.

\subsection{Boolean functions}

A Boolean function is a function of the form $f\colon \{0,1\}^n \to \{0,1\}$, where $n$ is a positive integer. The degree of $f$ is defined using the following well-known fact, see, e.g., \cite[Lemma 1]{survey_buhrman_2002}.

\begin{fact}[Multilinear representation]
   For every $f\colon \{0,1\}^n \to \mathbb{R}$, there exists a unique multilinear $p\in \mathbb{R}[X_1, \dots, X_n]$ such that $p(x) = f(x)$ for all $x\in \{0,1\}^n$.
\end{fact}

\begin{definition}[Degree]
    The \emph{degree} of $f\colon\{0,1\}^n\to \{0,1\}$, denoted $\deg(f)$, is defined to be the degree of the unique multilinear $p\in \mathbb{R}[X_1, \dots, X_n]$ such that $p(x) = f(x)$ for all $x\in \{0,1\}^n$.
\end{definition}

We now give the definitions of rational degree, sign degree and nondeterministic degree.

\begin{definition}[Rational degree]
    We say that $p/q\in \mathbb{R}(X_1,\dots,X_n)$, where $p,q$ are multilinear, is a \emph{rational representation} of $f\colon \{0,1\}^n \to \{0,1\}$ if for all $x\in \{0,1\}^n$, $q(x)\neq 0$ and $p(x)/q(x)=f(x)$. The \emph{rational degree} of $f$, denoted $\rdeg(f)$, is the minimum value of $\max(\deg(p),\deg(q))$ over all rational representations of $f$.
\end{definition}

\begin{definition}[Sign degree]
    We say that $p\in \mathbb{R}[X_1,\dots,X_n]$, where $p$ is multilinear, is a \emph{sign representation} of $f\colon \{0,1\}^n \to \{0,1\}$ if for all $x\in \{0,1\}^n$, $p(x)<0$ if and only if $f(x)=1$. The \emph{sign degree} of $f$, denoted $\sdeg(f)$, is the minimum value of $\deg(p)$ over all sign representations of $f$.
\end{definition}

\begin{definition}[Nondeterministic degree]
    We say that $p\in \mathbb{R}[X_1,\dots,X_n]$, where $p$ is multilinear, is a \emph{nondeterministic representation} of $f\colon \{0,1\}^n \to \{0,1\}$ if for all $x\in \{0,1\}^n$, $p(x)\neq 0$ if and only if $f(x)=1$. The \emph{nondeterministic degree} of $f$, denoted $\ndeg(f)$, is the minimum value of $\deg(p)$ over all nondeterministic representations of $f$.
\end{definition}

We record the following fact relating rational degree to sign degree mentioned in the introduction.

\begin{fact}\label{fact:sdeg_rdeg}
    For every Boolean function $f$, $\sdeg(f)/2 \leq \rdeg(f)$.\footnote{The factor of $2$ loss is unavoidable: for even $n$, $\sdeg(\PARITY_n) = n$ but $\rdeg(\PARITY_n) = n/2$.}
\end{fact}

Our proof will use the following characterization of rational degree in terms of nondeterministic degrees. The result is folklore and we give a proof for completeness.

\begin{fact}\label{fact:rdeg_ndeg}
    For every Boolean function $f$, $\rdeg(f)=\max(\ndeg(f),\ndeg(\neg f))$.
\end{fact}

\begin{proof}
    Let $p/q$ be a rational representation of $f$ with $\max(\deg(p),\deg(q))=\rdeg(f)$. Then, $p$ and $q-p$ are nondeterministic representations of $f$ and $\neg f$ respectively. Thus,
    \begin{equation*}
        \max(\ndeg(f),\ndeg(\neg f)) \leq \max(\deg(p),\deg(q-p)) \leq \max(\deg(p),\max(\deg(p),\deg(q))) =\rdeg(f).
    \end{equation*}
    
    Conversely, let $p$ and $q$ be nondeterministic representations of $f$ and $\neg f$ respectively, with $\deg(p)=\ndeg(f)$ and $\deg(q)=\ndeg(\neg f)$. Then $p/(p+q)$ is a rational representation of $f$. Thus,
    \begin{equation*}
        \rdeg(f)\leq\max(\deg(p),\deg(p+q))\leq \max(\deg(p),\max(\deg(p),\deg(q)))=\max(\ndeg(f),\ndeg(\neg f)).
    \end{equation*}
    
    Taken together, we obtain $\rdeg(f)=\max(\ndeg(f),\ndeg(\neg f))$, as required.
\end{proof}

We also record the fact that rational degree exactly equals the $\epsilon$-approximate postselected quantum query complexity $\PostQ_\epsilon$ with $\epsilon = 0$ as defined in \cite{postqe_mahadev_2015}. We will not use this fact later but it serves to underscore the tight connection between rational degree and quantum complexity, a connection first observed by Aaronson \cite{postselection_aaronson_2005}.

\begin{restatable}{fact}{rdegpostqe}
    For every Boolean function $f$, $\rdeg(f) = \PostQ_0(f)$.
\end{restatable}

\begin{proof}[Proof sketch]
    \cite{postqe_mahadev_2015} proves $\PostQ_\epsilon(f) \leq \rdeg(f) \leq 2 \PostQ_\epsilon(f)$ for all $\epsilon\in [0,1/2)$. When $\epsilon=0$, the factor of $2$ can be removed by extending the argument in \cite{postqe_mahadev_2015}; see \cref{app:rdeg_postqe}.
\end{proof}

Our proof will relate rational degree to degree via the following three complexity measures.

\begin{definition}[Block sensitivity]
    Let $f\colon \{0,1\}^n \to \{0,1\}$. Let $x\in \{0,1\}^n$. We say that a subset $B\subseteq [n]$ is a \emph{sensitive block} of $f$ at $x$ if $f(x)\neq f(x^B)$, where $x^B$ denotes $x$ with all bits in $B$ flipped. The \emph{block sensitivity} of $f$ at $x$, denoted $\bs_x(f)$, is the maximum number of disjoint sensitive blocks of $f$ at $x$.
\end{definition}

\begin{definition}[Hitting set]
    For a multilinear polynomial $p \coloneqq \sum_{S\subseteq[n]} c_S \cdot \prod_{i\in S} X_i\in \mathbb{R}[X_1,\dots,X_n]$, its \emph{set of maximal monomials} is the set $\calM(p) \coloneqq \{M \subseteq [n] \colon  \text{$c_M\neq 0$, $c_{M'}= 0$ for all $M'\supsetneq M$}\}$.\footnote{Note that this is not the same as $\calM'(p) \coloneqq \{M \subseteq [n] \colon  \text{$c_M\neq 0$ and $\deg(p) = \abs{M}$}\}$. While $\calM'(p)$ is always contained in $\calM(p)$, the containment could be strict. This distinction matters for \cref{sec:improved_lower_bound}.} We say that $H\subseteq [n]$ is a \emph{hitting set} of $p$ if $H \cap M$ is nonempty for all $M \in \calM(p)$.
\end{definition}

\begin{definition}[Decision tree complexity]
    The \emph{decision tree complexity} of $f\colon \{0,1\}^n \to \{0,1\}$, denoted $\D(f)$, is the minimum depth of a decision tree (deterministic query algorithm) that, for all $x\in \{0,1\}^n$, queries bits of $x$ to exactly compute $f(x)$ at a leaf.
\end{definition}

\subsection{Technical tools}

Our proof involves using polynomial symmetrization followed by Markov's inequality \cite{question_markov_1890}. The symmetrization we use is folklore, see, e.g., \cite[Lemma 12]{laurent_aaronson_2020} where it is attributed to \cite{apx_shi_2002}, and follows immediately from basic properties of expectation.

\begin{fact}\label{fact:bernoulli_symmetrization}
Given multilinear $p\in \mathbb{R}[X_1,\dots,X_n]$, define $P\in \mathbb{R}[Y]$ by replacing every monomial $\prod_{i\in S} X_i$ appearing in $p$ by $Y^{\abs{S}}$. For $y \in [0,1]$, let $\ber_{y}^n$ denote the distribution over $\{0,1\}^n$ where each bit is sampled independently to be $1$ with probability $y$. Then, for all $y \in [0,1]$,
\begin{equation}
    P(y) = \expect_{x \sim \ber_y^n}[ \, p(x)  \, ].
\end{equation}
Furthermore, $\deg(P) \leq \deg(p)$.
\end{fact}

\begin{theorem}[Markov]\label{thm:markov}
    Let $P \in \mathbb{R}[X]$. Let $a_1,a_2,b_1,b_2 \in \mathbb{R}$ be such that $a_1 < a_2$ and $b_1 < b_2$. Suppose $P(x) \in [b_1,b_2]$ for all $x\in [a_1,a_2]$. Then, for all $x\in [a_1,a_2]$,
    \begin{equation}
        \abs{P'(x)}~\leq~\frac{b_2-b_1}{a_2-a_1}\cdot \deg(P)^2.
    \end{equation}
\end{theorem}

\begin{corollary}\label{cor:approx}
    Let $p \in \mathbb{R}[X_1,\dots,X_n]$ and $h>0$. Suppose that $p$ has the following properties:
    \begin{enumerate}[label=(\roman*)]
        \item $\abs{p(x)}\leq h$ for all $x\in \{0,1\}^n$,
        \item $\abs{p(0^n)}=h$,
        \item $p(x) \cdot p(0^n) \leq 0$ for all $x \in \{0,1\}^n$ with $\abs{x} = 1$.
    \end{enumerate}
    Then,
    \begin{equation}
        \sqrt{n/2}\leq \deg(p).
    \end{equation}
\end{corollary}

\begin{proof}
    There are two cases to consider: (1) $p(0^n) = h $ and $p(x) \leq 0$ for all $x \in\{0,1\}^n$ with $\abs{x} = 1$; (2) $p(0^n) = -h$ and $p(x) \geq 0$ for all $x \in \{0,1\}^n$ with $\abs{x} = 1$. We give the proof for the first case as the second case then follows from considering $(-p)$.
    We may assume that $p$ is multilinear without loss of generality and write
    \begin{equation}
        p = a_0+(a_1 X_1+\cdots+a_n X_n)+ \textnormal{(higher degree terms)}.
    \end{equation}
    By assumption $a_0=p(0^n)=h$. Evaluating $p$ at bitstrings of Hamming weight 1, we see that $a_0+a_i \leq 0$ for all $i \in [n]$. Hence, $a_i \leq - h$ for all $i \in [n]$. By \cref{fact:bernoulli_symmetrization}, we obtain $P\in \mathbb{R}[Y]$ with $\deg(P)\leq \deg(p)$, such that, for all $y\in [0,1]$,
    \begin{equation}
        \abs{P(y)}\leq \expect_{x \sim \ber_y^n}[ \, \abs{p(x)} \, ]\leq h.
    \end{equation}
    Moreover, we can write
    \begin{equation}
        P=a_0 + (a_1 + \cdots + a_n)Y + \textnormal{(higher degree terms)}.
    \end{equation}
    In particular, we see that $P'(0) = a_1 + \cdots + a_n \leq  -n\cdot h$. Now \cref{thm:markov} gives
    \begin{equation}
        n\cdot h \leq \abs{P'(0)}\leq \frac{h-(-h)}{1-0}\cdot \deg(P)^2=2h\deg(P)^2.
    \end{equation}
    Therefore, $\sqrt{n/2}\leq \deg(P)\leq \deg(p)$, as required.
\end{proof}

The next key lemma relates sensitive blocks of $f$ at $x$ with maximal monomials of a nondeterministic representation of $f$. It generalizes \cite[Lemma 5]{survey_buhrman_2002} (attributed to Nisan and Smolensky) to involve $x$, nondeterministic representation, and maximal monomials. The proof is the same.
\begin{lemma}[Nisan-Smolensky]\label{lem:nisan-smolensky}
    Let $f\colon \{0,1\}^n \to \{0,1\}$ be nonconstant and $x \in f^{-1}(0)$. Let $p$ be a nondeterministic representation of $f$ and $M$ be a maximal monomial of $p$. Then there exists $B\subseteq M$ such that $f(x^B) = 1$.
\end{lemma}
\begin{proof}
Let $q$ be the restriction of $p$ obtained by fixing all variables outside $M$ according to $x$. Now $M$ remains a maximal monomial of $q$. Therefore, $q$ is nonconstant (as a formal polynomial). Therefore, there exists $y \in \{0,1\}^M$ with $q(y) \neq 0$ by uniqueness of multilinear representation. Let $x' \in \{0,1\}^n$ be equal to $y$ on $M$ and be equal to $x$ outside $M$. Then $p(x') = q(y) \neq 0$, which implies $f(x') = 1$ as $p$ is a nondeterministic representation of $f$. By definition, $x'$ and $x$ can only differ on $M$, so $x'$ can be written as $x^B$ for some $B\subseteq M$.
\end{proof}

The next two theorems are used in \cref{sec:implications}. The first, Minsky-Papert \cite{perceptrons_minsky_1969} symmetrization, is used to show the impossibility of generalizing our main result to all partial Boolean functions. The second, a lemma of Nisan and Szegedy from \cite{degree_nisan_1994}, which they attribute to Schwartz \cite{poly_schwartz_1980}, is used to show a lower bound on the rational degree of Boolean functions depending on $n$ variables. For both theorems we give a concise proof for completeness.

\begin{theorem}[Minsky-Papert]\label{thm:symmetrization}
For every $p\in \mathbb{R}[X_1,\dots, X_n]$, there exists $P\in \mathbb{R}[X]$ such that $\deg(P) \leq \deg(p)$ and, for all $i \in \{0,1,\dots,n\}$,
\begin{equation}
    P(i) = \binom{n}{i}^{-1} \sum_{x \in \{0,1\}^n \colon \abs{x} = i} p(x).
\end{equation}
\end{theorem}

\begin{proof}
    Assume $p$ is multilinear without loss of generality. For $i\in \{0,1,\dots,n\}$, and $d$ distinct indices $j_1,\dots,j_d \in [n]$, we have
    \begin{align*}
    \binom{n}{i}^{-1}\sum_{x \in \{0,1\}^n \colon \abs{x} = i}x_{j_1} \cdots x_{j_d} 
    = \binom{n}{i}^{-1} \binom{n-d}{i-d} = \frac{i(i-1)\cdots(i-d+1)}{n(n-1)\cdots (n-d+1)},
    \end{align*}
    which is a polynomial in $i$ of degree $d$. (Note that when $i<d$, it is the zero polynomial.) Since $p$ is a linear combination of monomials, each of degree at most $\deg(p)$, the theorem follows.
\end{proof}

\begin{theorem}[{\cite[Lemma 2.6]{degree_nisan_1994}}]\label{thm:ns94}
    Let $p\in \mathbb{R}[X_1,\dots,X_n]$ be a nonzero multilinear polynomial. Then, for $x$ chosen uniformly at random from $\{0,1\}^n$, it holds that $\Pr[p(x)\neq 0]\geq 2^{-\deg(p)}$.
\end{theorem}

\begin{proof}We give an alternative proof of \cref{thm:ns94}. If $\deg(p)=0$, then $p$ must be a nonzero constant, so assume $\deg(p)>0$. Fix a monomial $M$ of $p$ with $\abs{M} = \deg(p)$. For each of the $2^{n-\deg(p)}$ partial assignments to variables outside $M$, the restricted polynomial has degree $\deg(p)>0$. Therefore, by uniqueness of multilinear representation, this produces an $x\in \{0,1\}^n$ with $p(x) \neq 0$ that is consistent with each partial assignment. Therefore, $\Pr[p(x)\neq 0]\geq 2^{-n}\cdot 2^{n-\deg(p)}=2^{-\deg(p)}$.
\end{proof}

\section{Rational degree lower bound}
\label{sec:poly_lower_bound}

Key to our proof is the use of \emph{minimum} block sensitivity over elements of $f^{-1}(b)$ as an intermediary complexity measure. As far as we are aware, this measure has not been previously leveraged in the literature.\footnote{The closest measure that has appeared in the literature may be the minimum certificate complexity, see, e.g., \cite{composition_tal_2013,kill_odonnell_2014,fourier_arunachalam_2021}. We thank Ronald de Wolf for bringing this measure to our attention. This measure will serve to motivate the approach in \cref{sec:improved_lower_bound}.} The next two lemmas show how the measure can be used both as a lower bound and as (part of) an upper bound.

\begin{lemma}\label{lem:sdeg_lower}
    For every nonconstant $f \colon \{0,1\}^n \to \{0,1\}$,
    \begin{equation}\label{eq:sdeg_lower}
        \min_{x\in \{0,1\}^n}\bs_x(f)\leq 2\sdeg(f)^2.
    \end{equation}
    In particular, either
    \begin{equation}
        \min_{x\in f^{-1}(0)}\bs_x(f)\leq 2\sdeg(f)^2 \quad \text{or} \quad \min_{x\in f^{-1}(1)}\bs_x(f)\leq 2\sdeg(f)^2.
    \end{equation}
\end{lemma}

\begin{proof}
    Let $p$ be a sign representation of $f$ of degree $\sdeg(f)$. Let $h$ be the maximum value of $\abs{p(x)}$ over $x\in \{0,1\}^n$, which is strictly positive. Let $z\in \{0,1\}^n$ be such that $\abs{p(z)}=h$. Let $b \coloneqq \bs_z(f)$, and $B_1,\dots,B_b$ be disjoint sensitive blocks of $f$ at $z$. 
    
    Now consider the function $\tilde{r}\colon \{0,1\}^b \to \mathbb{R}$ defined by $\tilde{r}(t_1,\dots,t_b) = p(x_1,\dots,x_n)$, where
    \begin{align}
        x_j &= \begin{cases}
            z_j &\text{if } j\notin \bigsqcup_i B_i, \\ 
            t_i &\text{if } j\in B_i \text{ and } z_j=0,\\ 
            1 - t_i &\text{if } j\in B_i \text{ and } z_j = 1. 
        \end{cases}
    \end{align}
    By construction, $\tilde{r}$ can be represented on $\{0,1\}^b$ by a multilinear $r \in \mathbb{R}[X_1,\dots,X_b]$ of degree at most $\deg(p)$. Moreover, $\abs{r(u)} \leq h$ for all $u\in \{0,1\}^b$, $\abs{r(0^b)}=\abs{p(z)} = h$, and $r(u) \cdot r(0^b) \leq 0$ for all $u \in \{0,1\}^b$ with $\abs{u} = 1$. Therefore, \cref{cor:approx} gives
    \begin{equation}
        \sqrt{b/2}\leq \deg(r).
    \end{equation}
    Therefore,
    \begin{equation}
        \min_{x\in \{0,1\}^n} \bs_x(f)\leq b\leq 2\deg(r)^2\leq 2\deg(p)^2=2\sdeg(f)^2,
    \end{equation}
    as required.
\end{proof}

\begin{remark}\label{rem:sdeg_lower}
\cref{lem:sdeg_lower} is false if the $\min$ in \cref{eq:sdeg_lower} is replaced by $\max$ as witnessed by the majority function. Moreover, \cref{lem:sdeg_lower} is optimal up to an $O(\log n)$ factor as witnessed by the function $\OR_n \circ \AND_n$.
First, $\sdeg(\OR_n \circ \AND_n) = O(\sqrt{n \log n})$. We know a polynomial of degree $O(\sqrt{n \log(1/\epsilon)})$ for $\AND_n$ that evaluates to $1$ on $1$-inputs and evaluates within $[0,\epsilon]$ for $0$-inputs~\cite{inclusion_khan_1996,bounds_buhrman_1999}. Summing this polynomial for each of the $n$ $\AND_n$ gates that feed into the top $\OR_n$ gate with $\epsilon=1/(2n)$, we get a polynomial that is $\geq 1$ on $1$-inputs and $\leq 1/2$ on $0$-inputs, which can be made into a sign-representing polynomial by subtracting $3/4$ and then multiplying by $(-1)$. On the other hand, $\bs_x(\OR_n \circ \AND_n)=n$ for every input $x$, and hence $\min_x \bs_x(\OR_n \circ \AND_n)=n$.
\end{remark}

The next lemma is a tightened version of \cite[Lemma 6]{survey_buhrman_2002}, attributed to Nisan and Smolensky.

\begin{lemma}\label{lem:hitset_upper}
    Let $f$ be a nonconstant Boolean function. Let $p,q$ be nondeterministic representations of $f$ and $\neg f$, respectively. Then, both
    \begin{itemize}
        \item $p$ has a hitting set $H$ of size at most $ \abs{H} \leq \deg(p)\cdot \min_{x\in f^{-1}(0)}\bs_x(f)$, and
        \item $q$ has a hitting set $K$ of size at most $ \abs{K} \leq \deg(q)\cdot \min_{x\in f^{-1}(1)}\bs_x(f)$.
    \end{itemize}
\end{lemma}

\begin{proof}
    We will prove the lemma for $p$ as the proof is analogous for $q$. 
    
    Let $x\in f^{-1}(0)$. Let $\{M_1,\dots,M_b\}$ be a maximal set of disjoint maximal monomials of $p$. Then $p$ must have a hitting set of size at most $\deg(p) \cdot b$. This is because $H \coloneqq \bigsqcup_{i=1}^b M_i$ is a hitting set. Otherwise there is another maximal monomial that we can add to $\{M_1,\dots,M_b\}$. But $b\leq \bs_x(f)$ because each $M_i$ contains a sensitive block of $f$ at $x$ by \cref{lem:nisan-smolensky}. Therefore, we deduce 
    \begin{equation}
        \abs{H} \leq \deg(p) \cdot \bs_{x}(f).
    \end{equation}
    Since $x$ is an arbitrary element from $f^{-1}(0)$, the lemma for $p$ follows.
\end{proof}

Combining \cref{lem:sdeg_lower,lem:hitset_upper}, we obtain
\begin{corollary}\label{cor:hitset_upper_combined}
    Let $f$ be a nonconstant Boolean function. Let $p, q$ be nondeterministic representations of $f$ and $\neg f$, respectively. Then, either
    \begin{itemize}
        \item $p$ has a hitting set of size at most $2\deg(p)\sdeg(f)^2$, or
        \item $q$ has a hitting set of size at most $2\deg(q)\sdeg(f)^2$.
    \end{itemize}
\end{corollary}

From \cref{cor:hitset_upper_combined}, we obtain our main theorem by explicitly constructing a decision tree for $f$ using nondeterministic representations of $f$ and $\neg f$. Since $\deg(f) \leq \D(f)$, see, e.g., \cite[Theorem 10]{survey_buhrman_2002}, the theorem implies $\deg(f) \leq 16 \rdeg(f)^4$.

\begin{theorem}\label{thm:D_upper}
    For every Boolean function $f$,
    \begin{equation}
        \D(f) \leq 4 \sdeg(f)^2 \, \rdeg(f)^2 \leq 16 \rdeg(f)^4.
    \end{equation}
\end{theorem}

\begin{proof}
    The second inequality follows from \cref{fact:sdeg_rdeg} so it suffices to prove the first. Let $p$ and $q$ be nondeterministic representations of $f$ and $\neg f$, respectively, such that $\deg(p)=\ndeg(f)$ and $\deg(q)=\ndeg(\neg f)$.
    
    We claim that the following \cref{alg:determining-f-one} gives a deterministic query algorithm that computes $f$ and uses at most $4 \sdeg(f)^2 \, \rdeg(f)^2$ queries.

    \begin{algorithm}[H]
    \caption{Deterministic query algorithm for computing $f$}\label{alg:determining-f-one}
    \begin{algorithmic}[1]
    \State $i\leftarrow 1$; $(f^1, p^1, q^1) \leftarrow (f, p, q)$
    \While{$f^i$ is not constant} \label{ln:alg-1-while}
        \State Query a hitting set of $p^i$ of size at most $2 \deg(p^i) \sdeg(f^i)^2$ if it exists, else query a hitting set of $q^i$ of size at most $2 \deg(q^i) \sdeg(f^i)^2$. \label{ln:alg-1-hit}
        \State Set $f^{i+1}$, $p^{i+1}$, $q^{i+1}$ to be $f^i$, $p^i$, $q^i$, respectively, but with the variables just queried fixed to their queried values.
        \State $i\leftarrow i+1$
    \EndWhile
    \State \Return constant value of $f^i$
    \end{algorithmic}
    \end{algorithm}

    Since restriction preserves nondeterministic representations, at every iteration of the while loop, $p^i$ is a nondeterministic representation of $f^i$ and  $q^i$ is a nondeterministic representation of $\neg f^i$. Therefore, \cref{ln:alg-1-hit} is possible by \cref{cor:hitset_upper_combined}.

    First note that the value the algorithm returns must be the value of $f$ on the input. Indeed, $f^i$ is a restriction of $f$ to values queried from the input, and the condition of the while loop ensures that $f^i$ is a constant upon termination.

    Next we bound the number of queries the algorithm uses. By the definition of a hitting set, if a hitting set of $p^i$ is queried, then the degree of $p^{i+1}$ must be strictly less than that of $p^i$. On the other hand, if a hitting set of $q^i$ is queried, then the degree of $q^{i+1}$ must be strictly less than that of $q^i$. If either $p^i$ or $q^i$ is constant, $f^i$ is constant since $p^i,q^i$ are nondeterministic representations of $f^i$ and $\neg f^i$, respectively. Therefore, the number of iterations that make queries is at most $\deg(p)+\deg(q)$, which is at most $2\rdeg(f)$.
    
    Moreover, the number of queries at the $i$th iteration is at most 
    \begin{align*}
        &~2 \max(\deg(p^i),\deg(q^i)) \cdot \sdeg(f^i)^2
        \\
        ~\leq&~2\max(\deg(p),\deg(q)) \cdot \sdeg(f^i)^2
        &&\text{($p^i,q^i$ are restrictions of $p,q$)}
        \\
        ~=&~2\rdeg(f) \sdeg(f^i)^2
        &&\text{(\cref{fact:rdeg_ndeg})}
        \\
        ~\leq&~2 \rdeg(f) \sdeg(f)^2
        &&\text{($f^i$ is a restriction of $f$)}.
    \end{align*}

    Therefore, the total number of queries is at most $4\sdeg(f)^2 \rdeg(f)^2$, as required.
\end{proof}

\section{Improved lower bound}
\label{sec:improved_lower_bound}

To understand how the previous proof can be improved, let us recap it in a more modular way. For this we need the notion of \emph{certificate complexity}. Given $f\colon \{0,1\}^n \to \{0,1\}$ and $x\in \{0,1\}^n$, the certificate complexity of $f$ at $x$, denoted $\C_x(f)$, is the minimum number of bits of $x$ such that fixing them determines $f(x)$.

\paragraph{Step 1.} We start with the argument that upper bounds $\D(f)$ in the proof of \cref{thm:D_upper}. The idea there is that if $p$ is a nondeterministic polynomial for $f$ and we query a hitting set for $p$, then restricting $p$ to the queried variables produces a polynomial that nondeterministically represents the restricted function, and its degree has dropped by 1 (or more). One simple hitting set for any nondeterministic polynomial for $f$ is any $0$-certificate for $f$. To see this, note that a $0$-certificate must intersect with every maximal monomial of a nondeterministic polynomial $p$. If it did not, then by restricting to those variables we get a function that cannot be identically zero (since it has at least one monomial), so there is some setting of those bits that makes it evaluate to $1$. (This is the argument proving \cref{lem:nisan-smolensky}.) But then this $1$-input is consistent with the $0$-certificate, which is a contradiction. This argument shows that querying any $0$-certificate reduces the degree of our nondeterministic polynomial by $1$. More generally, if we pick the input $x$ with minimum certificate complexity, $\Cmin(f)=\min_x \C_x(f)$, this reduces the degree of either the nondeterministic polynomial for $f$ or $\neg f$ by $1$, as in the current argument. 

It is tempting to now conclude (incorrectly) that $\D(f) \leq (\ndeg(f)+\ndeg(\neg f)) \Cmin(f)$ because each query of size $\Cmin(f)$ drops the degree of one of these polynomials by at least $1$, and hence after $\ndeg(f)+\ndeg(\neg f)$ rounds, we arrive at a trivial polynomial. The flaw in this argument is that after the first step, we have a restriction $g$ of $f$, and it is not always true that $\Cmin(f) \geq \Cmin(g)$. This generally happens with all ``best-case'' complexity measures (i.e., ones where we minimize over inputs instead of maximizing over inputs as done in traditional worst-case measures), because a restriction might kill the best input making the complexity measure larger. The standard solution for such best-case measures is to explicitly define a ``downward-closed'' version of the measure that maximizes over all restrictions. More formally, we define $\Cmindown(f)=\max_g \Cmin(g)$, where $g$ ranges over all restrictions of $f$. We can now (correctly) conclude that
\begin{equation}\label{eq:DCmindown}
    \D(f) \leq \bigl(\ndeg(f)+\ndeg(\neg f)\bigr) \Cmindown(f) \leq 2\rdeg(f) \Cmindown(f).
\end{equation}

\paragraph{Step 2.} We now turn to the argument in \cref{lem:hitset_upper} that gives us an upper bound on $\Cmindown(f)$. The standard argument (due to Nisan~\cite{crew_nisan_1991}) that relates $\C(f)$ and $\bs(f)$ says that for every input $x$, if we take a set of maximal disjoint sensitive blocks, their union is a certificate for $x$. We know that every input $x$ can have at most $\bs_x(f)$ disjoint sensitive blocks, and each of them can be taken to be a minimal sensitive block, so all we need to do is upper bound the size of the largest minimal sensitive block. A minimal sensitive block of size $k$ on a $0$-input is just a copy of the $\AND$ function of size $k$ hiding inside our function up to negating the input variables. For example, on the $0^n$ input, a minimal sensitive block of size $k$ has the property that it evaluates to $0$ when at most $k-1$ of its bits are set to $1$ and evaluates to $1$ when all $k$ bits are set to $1$, which is exactly the $\AND_k$ function. Similarly for a $1$-input we get the $\neg \AND_k$ function. In Nisan's standard argument, we upper bound the size of the largest minimal block by $\mathrm{s}(f)$, the sensitivity of $f$, since the $\AND$ and $\neg\AND$ functions of size $k$ have sensitivity $k$. But we can also upper bound it by $\rdeg(f)$, since the $\AND$ and $\neg\AND$ functions of size $k$ have rational degree $k$~\cite{rdeg_iyer_2025}. This gives us $\C_x(f) \leq \bs_x(f) \rdeg(f)$. By minimizing over all inputs, we get
\begin{equation}
    \Cmin(f) \leq \bsmin(f) \rdeg(f).
\end{equation}

\paragraph{Step 3.} Finally, \cref{lem:sdeg_lower} already gives us the last part of the argument:
\begin{equation}
    \bsmin(f) \leq 2 \sdeg(f)^2.
\end{equation}

Combining steps 2 and 3 gives us $\Cmin(f) \leq 2 \sdeg(f)^2 \rdeg(f)$. To use Step 1, we need to upper bound $\Cmindown(f)$, not $\Cmin(f)$. But degree measures do not increase under downward closure, so we immediately get 
\begin{equation}
    \Cmindown(f) \leq 2 \sdeg(f)^2 \rdeg(f).
\end{equation}
Combining this with \cref{eq:DCmindown} gives us $\D(f) \leq 4 \sdeg(f)^2 \rdeg(f)^2$.

\paragraph{Slack in the proof.} So how can this proof be tightened? Step 1 upper bounds $\D(f)$ using minimum certificate complexity and Step 3 lower bounds sign degree (squared) using minimum block sensitivity, and we know that $\Cmin(f)$ and $\bsmin(f)$ can genuinely be different, just like $\C(f)$ and $\bs(f)$ can be different. Can we find a complexity measure that is in between certificate complexity and block sensitivity as a compromise?

In previous work~\cite{composition_tal_2013,fbs_gilmer_2016}, it was observed that for every input $x$, $\C_x(f)$ and $\bs_x(f)$ admit integer-program formulations whose linear-programming (LP) relaxations are dual to each other. Consequently, the optimum values of these relaxations are equal by strong duality. The LP relaxation of $\bs_x(f)$ is called fractional block sensitivity and is denoted $\fbs_x(f)$. The LP relaxation of $\C_x(f)$ is called fractional certificate complexity, but we will instead use an algorithmic interpretation of it called randomized certificate complexity, denoted $\RC_x(f)$, which was first defined in \cite{qcertificate_aaronson_2008}.

A clear strategy emerges that replaces the previous three-step strategy:
\begin{enumerate}
    \item[Step 1.] Upgrade $\D(f) \leq 2 \rdeg(f) \Cmindown(f)$ to have $\RCmindown(f)$ instead of $\Cmindown(f)$.
    \item[Step 2.] Use linear programming duality to conclude that $\RCmindown(f) = O(\fbsmindown(f))$.
    \item[Step 3.] Upgrade $\bsmin(f) \leq 2 \sdeg(f)^2$ to have $\fbsmin(f)$ instead of $\bsmin(f)$.
\end{enumerate}

This strategy conceptually mirrors the previous proof. In Step 1, we upper bound $\D(f)$ using $\RC$, which is morally quite similar to $\C$. Informally, both are algorithmic measures, in the sense that it is easy to upper bound them by exhibiting a (deterministic or randomized) certificate. In Step 3, we replace $\bs$ with $\fbs$, which are also morally similar, since they are both lower bound measures: it is easy to lower bound them by exhibiting an (integer or fractional) set of sensitive blocks. The old Step 2 used Nisan's argument to relate an algorithmic measure ($\C$) to a lower bound measure ($\bs$). In the new Step 2, we use linear programming duality to relate the algorithmic measure ($\RC$) to the lower bound measure ($\fbs$).

We execute this strategy for the remainder of this section. Our upper bound on $\D(f)$ contains an extra $\log n$ factor, which is why the overall result we obtain is $\D(f) \leq O(\rdeg(f)^3 \log n)$.

\subsection{Randomized certificate complexity and fractional block sensitivity}

The definition of randomized certificate complexity was first introduced in \cite{qcertificate_aaronson_2008}. The definition uses the notion of \emph{randomized decision tree}, which is a probability distribution over decision trees.
\begin{definition}[Randomized certificate complexity]
Let $f\colon \{0,1\}^n \to \{0,1\}$. Let $x\in \{0,1\}^n$. We say that a randomized decision tree $\tau$ is an $\rc$-verifier for $x$ (with respect to $f$) if (i) given input $x$, $\tau$ accepts with probability $1$; and (ii) given input $y\in \{0,1\}^n$ such that $f(y) \neq f(x)$, $\tau$ rejects with probability $\geq 1/2$. The randomized certificate complexity of $f$ at $x$, denoted $\rc_x(f)$, is defined to be the minimum expected number of queries made by an $\rc$-verifier for $x$.
\end{definition}

We observe that there always exists an $\rc$-verifier for $x$ that works by nonadaptively querying at most $O(\rc_x(f))$ bits of the input. This will be important for bounding $\D(f)$.
\begin{lemma}\label{lem:nonadaptive_rc}
    Let $f\colon \{0,1\}^n \to \{0,1\}$. Let $x\in \{0,1\}^n$. There exists an $\rc$-verifier $\tau$ for $x$ that works by sampling a random subset $Q\subseteq [n]$ of size $O(\rc_x(f))$, querying the input on $Q$, and accepting if and only if the returned values are consistent with $x$.
\end{lemma}

\begin{proof}
Let $\sigma$ be an $\rc$-verifier for $x$ whose expected number of queries is equal to $\rc_x(f)$. By definition, $\sigma$ is a randomized decision tree, meaning it is a probability distribution $(p_T)_T$ over decision trees $T$. For each decision tree $T$ with $p_T>0$, let $Q_T$ denote the set of indices that $T$ queries given input $x$. Since $\sigma$ accepts $x$ with probability $1$, $T$ must accept $x$.

Consider the following randomized decision tree $\tau''$: sample $Q_T$ with probability $p_T$, query the input on $Q_T$, and accept if and only if the returned values are consistent with $x$. Given input $x$, it is clear that $\tau''$ always accepts. Given input $y\in \{0,1\}^n$ such that $f(y) \neq f(x)$,
\begin{equation}
     \Pr[\text{$\tau''$ accepts $y$}] = \Pr_T[\text{$y$ agrees with $x$ on $Q_T$}] \leq \Pr_T[\text{$T$ accepts $y$}] \leq 1/2.
\end{equation}

Since $\expect_T[\abs{Q_T}] \leq \rc_x(f)$, Markov's inequality gives $\Pr_T[\abs{Q_T} \geq 10 \rc_x(f)] \leq 1/10$. Let $\tau'$ be the variant of $\tau''$ such that if $Q_T$ with $\abs{Q_T} \geq 10 \rc_x(f)$ is sampled, then $\tau'$ immediately accepts without making any queries. Given input $x$, it is clear that $\tau'$ always accepts. Given input $y\in \{0,1\}^n$ such that $f(y) \neq f(x)$,
\begin{equation}
     \Pr[\text{$\tau'$ accepts $y$}] \leq 1/2 + 1/10.
\end{equation}

Let $\tau$ be defined by repeating $\tau'$ twice and accepting if and only if $\tau'$ accepts both times. Given input $x$, it is clear that $\tau$ always accepts. Given input $y\in \{0,1\}^n$ such that $f(y) \neq f(x)$,
\begin{equation}
     \Pr[\text{$\tau$ accepts $y$}] \leq (1/2 + 1/10)^2 = 0.36 < 1/2.
\end{equation}

Therefore, $\tau$ is an $\rc$-verifier for $x$ that samples a random subset $Q\subseteq [n]$ with $\abs{Q}<20\rc_x(f)$, queries the input on $Q$, and accepts if and only if the returned values are consistent with $x$.
\end{proof}

The definition of fractional block sensitivity was first introduced in \cite{composition_tal_2013,fbs_gilmer_2016}.
\begin{definition}[Fractional block sensitivity]
    Let $f\colon \{0,1\}^n \to \{0,1\}$. Let $x\in \{0,1\}^n$. Write $\calB$ for the set of sensitive blocks of $f$ at $x$. The fractional block sensitivity of $f$ at $x$, denoted $\fbs_x(f)$, is defined to be the optimal value of the following linear program: 
    \begin{align*}
        \max \quad & \sum_{B \in \calB} w_B
        \\
        \mathrm{s.t.} \quad &\forall B \in \calB \colon\ w_B \geq 0
        \\
        \quad & \forall i \in [n] \colon \sum_{B \in \calB \colon i \in B} w_B \leq 1
    \end{align*}
    (If the $w_B$s are restricted to take values in $\{0,1\}$, then the value of the program equals $\bs_x(f)$.)
\end{definition}

The following tight relation between $\fbs_x(f)$ and $\rc_x(f)$ was first shown in \cite{composition_tal_2013,fbs_gilmer_2016}.

\begin{theorem}\label{thm:duality}
For every $f\colon \{0,1\}^n \to \{0,1\}$ and $x\in \{0,1\}^n$,  $\fbs_x(f) = \Theta(\rc_x(f))$.
\end{theorem}

\begin{proof}
    \cite[Theorem 2]{composition_tal_2013} shows $\fbs_x(f) = \mathrm{FC}_x(f)$ by linear programming duality, where $\mathrm{FC}_x(f)$ is the fractional certificate complexity of $f$ at $x$. But $\mathrm{FC}_x(f) = \Theta(\rc_x(f))$ by \cite[Claim 5.5]{composition_tal_2013}.
\end{proof}

To effectively use $\rc_x(f)$ in our proof, it is crucial to consider its minimum, or ``best-case'' value, over all $x\in \{0,1\}^n$. This is similar to how the use of $\min_{x\in \{0,1\}^n} \bs_x(f)$ is crucial in \cref{sec:poly_lower_bound}. We define this notion generally below.
\begin{definition}[Best-case measures]
    Let $f\colon \{0,1\}^n \to \{0,1\}$ and $x\in \{0,1\}^n$. Let $M_x(f)$ be a local measure of the complexity of $f$ at $x$. Then we define $M_{\min}(f) \coloneqq \min_{x\in \{0,1\}^n} M_x(f)$.
\end{definition}

We will also use downward closure, a notion which has appeared in works such as  \cite{sensitivity_linzhang_2017, diameter_chaubalgal_2021}.

\begin{definition}[Downward closure]\label{def:downward_closure}
Let $M$ be a complexity measure of Boolean functions. The downward closure of $M$, denoted $M^{\downarrow}$, is the complexity measure of Boolean functions defined by $M^{\downarrow}(f) \coloneqq \max_{\rho} M(f|_{\rho})$, where the maximum is taken over all restrictions $\rho$.
\end{definition}

With the above definitions in place, we can state the following corollary of \cref{thm:duality}.

\begin{corollary}\label{eq:rcmin_fbs}
    For every Boolean function $f$,  $\rc_{\min}^{\downarrow}(f) \leq O(\fbs^{\downarrow}_{\min}(f))$.
\end{corollary}

\begin{proof}
Let $\rho$ be a restriction of $f$ and $x$ be an input of $f|_{\rho}$. \cref{thm:duality} shows $\rc_{x}(f|_\rho) \leq O(\fbs_{x}(f|_\rho))$. Taking the minimum over $x$ gives $\rc_{\min}(f|_\rho) \leq O(\fbs_{\min}(f|_\rho))$. Taking the maximum over $\rho$ gives $\rc_{\min}^{\downarrow}(f) \leq O(\fbs^{\downarrow}_{\min}(f))$, as required.
\end{proof}

\subsection{Randomized certificate complexity upper bounds decision tree complexity}

We establish the following result in this section,
which strengthens \cite[Theorem 10]{qcertificate_aaronson_2008}\footnote{The original theorem shows $\R_0(f) \leq O\bigl(\ndeg(f) \cdot \max_{x\in \{0,1\}^n} \RC_x(f) \cdot \log n\bigr)$ for every $f\colon \{0,1\}^n \to \{0,1\}$, where $\R_0(f)$ is the zero-error randomized query complexity of $f$.} to involve $\rc_{\min}^{\downarrow}$ and $\D$. In order to involve $\rc_{\min}^{\downarrow}$, we exploit the freedom to choose which $\rc$-verifier to run in the zero-error randomized algorithm used to prove the original theorem. To involve $\D$, we further derandomize that algorithm. 

\begin{theorem}\label{thm:rcmin}
    For every $f\colon \{0,1\}^n \to \{0,1\}$, 
    \begin{equation}
        \D(f) \leq O\bigl(\rdeg(f) \, \rc_{\mathrm{min}}^{\downarrow}(f) \, \log{n}\bigr).
    \end{equation}
\end{theorem}

We will use the following two lemmas in the proof. For an $\rc$-verifier $\tau$ of the type in \cref{lem:nonadaptive_rc}, we will identify $\tau$ with the distribution from which it samples the subset it queries.

\begin{lemma}\label{lem:rc-verifier-hits-maximal-monomials}
    Let $g$ be a nonconstant Boolean function and $x \in g^{-1}(0)$. Let $p$ be a nondeterministic representation of $g$ and $M$ be a maximal monomial of $p$. Suppose $\tau$ is an $\rc$-verifier for $x$ of the type in \cref{lem:nonadaptive_rc}, then $\Pr\limits_{Q\sim \tau}[Q \cap M \neq \emptyset] \geq 1/2$.
\end{lemma}

\begin{proof}
    By \cref{lem:nisan-smolensky}, there exists $B \subseteq M$ such that $g(x^B) = 1$. By definition, $\tau$ accepts $x$ with probability $1$. Therefore, if $\tau$ rejects $x^B$, then it must query some of the variables in $B$. Therefore,
    \begin{equation}
        \Pr\limits_{Q\sim \tau}[Q\cap M \neq \emptyset] \geq \Pr\limits_{Q\sim \tau}[Q \cap B \neq \emptyset] \geq \Pr[\tau \text{ rejects $x^B$}] \geq \frac{1}{2},
    \end{equation}
    as required.
\end{proof}

For the next lemma, we need the following definition of a potential function, which is inspired by a similar function used to prove \cite[Lemma 9]{qcertificate_aaronson_2008}. This function will serve to track the progress made by the deterministic query algorithm we construct to prove \cref{thm:rcmin}.
For a multilinear polynomial $r$, recall that $\calM(r)$ denotes the set of maximal monomials of $r$.

\begin{definition}[Potential function]
    For multilinear $r\in \mathbb{R}[X_1,\dots,X_m]$, define
    \begin{equation}
        \Phi(r) \coloneqq \sum_{\emptyset \neq M \in \calM(r)} 3^{\lvert M \rvert} \lvert M \rvert!
    \end{equation}
\end{definition}

\begin{lemma}\label{cor:deterministic-rc-round} Let $g\colon \{0,1\}^m \to \{0,1\}$ be nonconstant. Let $r$ be a nondeterministic representation of $g$. Then for every $Q\subseteq[m]$ and every restriction $r'$ of $r$ obtained by fixing the variables in $Q$ to arbitrary values, $\Phi(r')\leq \Phi(r)$. Furthermore, for every $x \in g^{-1}(0)$,
there exists $Q^{*} \subseteq [m]$ with $\abs{Q^*} \leq O(\rc_x(g))$ such that, for every restriction $r'$ of $r$ obtained by fixing the variables in $Q^{*}$ to arbitrary values, $\Phi(r') \leq (3/4) \Phi(r)$.
\end{lemma}

\begin{proof}
    Let $Q\subseteq [m]$ and $r'$ be the restriction of $r$ obtained by fixing the variables in $Q$ to arbitrary values in $\{0,1\}$. It will be convenient to write
    \begin{equation}
        \Phi^{Q}_{\hit}(r) \coloneqq \sum_{\substack{\emptyset \neq M \in \calM(r) \colon \\ M \cap Q \neq \emptyset}} 3^{\lvert M \rvert} \lvert M \rvert! \quad \text{and} \quad \Phi^{Q}_{\unhit}(r) \coloneqq \sum_{\substack{\emptyset \neq M \in \calM(r)\colon\\ M \cap Q = \emptyset}} 3^{\lvert M \rvert} \lvert M \rvert!
    \end{equation}
    These definitions imply $\Phi(r) = \Phi^{Q}_{\hit}(r)+ \Phi^{Q}_{\unhit}(r)$.

    Consider a maximal monomial $M$ of $r$. If $M$ does not intersect $Q$, then $M$ remains a maximal monomial of $r'$. On the other hand, if $M$ does intersect $Q$, then $M$ cannot be present in $r'$ but it is possible for some submonomials of $M$ to become maximal monomials of $r'$.\footnote{For example, consider $r = X_1 + X_1 X_2$ and fix $X_2$ to $0$ to obtain $r' = X_1$.} Writing $k\coloneqq \abs{M}$, the contribution of such submonomials to $\Phi(r')$ is upper bounded by
    \begin{equation}
        \sum_{l = 1}^{k-1} \binom{k}{l} 3^l l! = 3^k k! \sum_{r=1}^{k-1} \frac{1}{3^r r!} \leq 3^k k! \sum_{r = 1}^{\infty} \frac{1}{3^r} = \frac{1}{2} \cdot 3^k k!
    \end{equation}
     Since $3^k k!$ is the contribution of $M$ to $\Phi^{Q}_{\hit}(r)$, we can sum over all maximal monomials $M$ of $r$ that intersect $Q$ to deduce
    \begin{equation}\label{eq:potential_decrease}
        \Phi(r') \leq \Phi^{Q}_{\unhit}(r) + \frac{1}{2} \Phi^{Q}_{\hit}(r) = \Phi(r) - \frac{1}{2} \Phi^{Q}_{\hit}(r).
    \end{equation}
    In particular, $\Phi(r')\leq \Phi(r)$. 

    We now proceed to prove the ``furthermore'' part. By \cref{lem:nonadaptive_rc}, there exists an $\rc$-verifier $\tau$ for $x$ that works by sampling a random subset $Q\subseteq[m]$ of size $O(\rc_x(g))$ and querying the input on $Q$. By \cref{lem:rc-verifier-hits-maximal-monomials}, $\Pr_{Q\sim \tau}[M \cap Q \neq \emptyset] \geq 1/2$ for every maximal monomial $M$ of $r$. Therefore,
    \begin{align}
        \expect_{Q\sim \tau}[\Phi^{Q}_{\hit}(r)] \geq \frac{1}{2} \Phi(r).
    \end{align}
    
    In particular, there exists $Q^{*} \in \supp(\tau)$ such that $\Phi_{\hit}^{Q^{*}}(r) \geq \Phi(r)/2$. As $Q^{*} \in \supp(\tau)$, we have $\abs{Q^*} \leq O(\rc_x(g))$. Moreover, \cref{eq:potential_decrease} shows that for every restriction $r'$ of $r$ obtained by fixing variables in $Q^*$, we have $\Phi(r') \leq  \Phi(r) - \Phi^{Q^{*}}_{\hit}(r)/2 \leq (3/4) \Phi(r)$, as required.
\end{proof}

We are now ready to prove \cref{thm:rcmin}.
\begin{proof}[Proof of \cref{thm:rcmin}] 

Let $p$ and $q$ be nondeterministic representations of $f$ and $\neg f$, respectively, such that $\deg(p)~=~\ndeg(f)$ and $\deg(q)~=~\ndeg(\neg f)$. For a nonconstant Boolean function $g$, an input $x\in g^{-1}(0)$, and a nondeterministic representation $r$ of $g$, denote by $\textsc{RC-Sets}(g, x, r)$ the sets $Q^*$ satisfying the conditions of \cref{cor:deterministic-rc-round} with respect to $(g,x,r)$.

We claim that the following \cref{alg:determining-f} gives a deterministic query algorithm that computes $f$ and uses at most $O(\rdeg(f) \, \rc_{\mathrm{min}}^{\downarrow}(f) \, \log{n})$ queries.

    \begin{algorithm}[H]
    \caption{Deterministic query algorithm for computing $f$}\label{alg:determining-f}
    \begin{algorithmic}[1]
    \State $\rho \leftarrow \text{$\emptyset$ (empty restriction)}$ \label{ln:let-rho-be-empty}
    \While{$f|_{\rho}$ is not constant} \label{ln:while}
        \State \Let $x \in \argmin_{y} \rc_y(f|_{\rho})$ \label{ln:let-g-let-x}
        \If{$f|_{\rho}(x) = 0$} \label{ln:if-gx-zero}
            \State \Let $Q^{*} \in \textsc{RC-Sets}(f|_{\rho}, x, p|_{\rho})$  \label{ln:pick-rc-set-for-g}
        \Else
            \State \Let $Q^{*} \in \textsc{RC-Sets}(\neg f|_{\rho}, x, q|_{\rho})$ \label{ln:pick-rc-set-for-not-g}
        \EndIf
        \State Query $Q^*$ and extend $\rho$ according to the queried values \label{ln:query-Q-star}
    \EndWhile
    
    \State \Return constant value of $f|_\rho$
    
    \end{algorithmic}
    \end{algorithm}

Since restriction preserves nondeterministic representations, at every iteration of the while loop, $p|_{\rho}$ and $q|_{\rho}$ are nondeterministic representations of $f|_{\rho}$ and $\neg f|_{\rho}$, respectively. Therefore, \cref{cor:deterministic-rc-round} guarantees that $\textsc{RC-Sets}(f|_{\rho}, x, p|_{\rho})$ and $\textsc{RC-Sets}(\neg f|_{\rho}, x, q|_{\rho})$ are nonempty.

First note that the value the algorithm returns must be the value of $f$ on the input. Indeed, $f|_\rho$ is a restriction of $f$ to values queried from the input, and the condition of the while loop ensures that $f|_\rho$ is a constant upon termination.

Next we bound the number of queries the algorithm uses. \cref{cor:deterministic-rc-round} shows that each iteration makes at most $O(\rc_{\min}^{\downarrow}(f))$ queries. Therefore, it suffices to bound the number of iterations by $O(\rdeg(f) \, \log n)$. We do so by tracking the following potential measure
\begin{equation}
    \Phi_{\rho} \coloneqq \Phi(p|_{\rho}) \cdot \Phi(q|_{\rho}),
\end{equation}
which satisfies the following three properties:
\begin{enumerate}
    \item Initially, we have $\Phi_{\emptyset} \leq n^{O(\rdeg(f))}$. This is because we can bound the potential of $p$ by the total potential of all possible monomials that could appear in a polynomial of degree $\deg(p) = \ndeg(f)$:
    \begin{equation}
        \Phi(p) \leq \sum_{k=1}^{\ndeg(f)} \binom{n}{k}3^k k! \leq \sum_{k=1}^{\ndeg(f)} (3n)^k \leq n^{O(\ndeg(f))},
    \end{equation}
    and, similarly,
    \begin{equation}
        \Phi(q) \leq n^{O(\ndeg(\neg f))}.
    \end{equation}

    Therefore,
    \begin{equation}
        \Phi_{\emptyset} = \Phi(p) \cdot \Phi(q) \leq n^{O(\ndeg(f) + \ndeg(\neg f))} \leq n^{O(\rdeg(f))}.
    \end{equation}
    
    \item If $\Phi_\rho < 1$, then $f|_\rho$ is constant. This is because $\Phi_\rho < 1$ implies $\Phi(p|_\rho)<1$ or $\Phi(q|_\rho)<1$. If $\Phi(p|_\rho)<1$, then $p|_\rho$ cannot have any nonempty maximal monomials as can be seen from the definition of $\Phi$. Therefore, $p|_\rho$ is constant and so $f|_\rho$ is constant. Similarly, if $\Phi(q|_\rho)<1$, then $f|_\rho$ is constant.
    \item $\Phi_\rho$ decreases by a factor of at least $3/4$ at every iteration. This is because \cref{cor:deterministic-rc-round} shows that one of $\Phi(p|_{\rho})$ or $\Phi(q|_{\rho})$ decreases by a factor of at least $3/4$, and the other cannot increase.
\end{enumerate}

Taken together, these three properties imply that the total number of iterations can be bounded by $\log_{4/3}(n^{O(\rdeg(f))}) \leq O(\rdeg(f) \log n)$, as required.
\end{proof}

\subsection{Fractional block sensitivity lower bounds sign degree}

In this subsection, we upgrade \cref{lem:sdeg_lower} to a lower bound on sign degree by minimum fractional block sensitivity. We implement the upgrade by adapting the proof of \cite[Lemma 28]{fbs_shalev_2021}.

\begin{lemma}[Bounded polynomial for partial-OR]\label{lem:approx_or}
    For every positive integer $k$, there exists a $q\in \mathbb{R}[X_1,\dots,X_k]$ of degree $\lceil (\pi/2) \cdot \sqrt{k}\rceil$ such that $q(0^k)=0$ and $q(e_j)=1$ for all $j\in [k]$ and $q(x)\in [0,1]$ for all $x\in \{0,1\}^k$.
\end{lemma}

\begin{proof}
    Let $d \coloneqq \lceil (\pi/2) \cdot \sqrt{k}\rceil$. Let $T_d$ be the degree-$d$ Chebyshev polynomial of the first kind.  Then, define $r\in \mathbb{R}[Z]$ and $q\in \mathbb{R}[X_1,\dots,X_k]$ by
    \begin{equation}
        r \coloneqq \frac{1-T_d(1-(1-\cos(\pi/d))Z)}{2} \quad \text{and} \quad q \coloneqq r(X_1+\cdots+X_k).
    \end{equation}

    By definition, $T_d(\cos \theta) = \cos(d\theta)$ for all real $\theta$. Therefore, $r(0)=0$ and $r(1)=1$, so $q(0^k)=0$ and $q(e_j)=1$ for all $j\in [k]$. Moreover, $\abs{T_d(x)}\leq 1$ for all $x\in [-1,1]$, so $r(z) \in [0,1]$ for all $z\in [0, 2/(1-\cos(\pi/d))]$. Using $1-\cos(\theta)\leq \theta^2/2$ for all real $\theta$, we deduce
    \begin{equation}
        \frac{2}{1-\cos(\pi/d)}\geq \frac{4d^2}{\pi^2}\geq k.
    \end{equation}
    Hence $r(z) \in [0,1]$ for all $z\in [0,k]$ and therefore $q(x)\in [0,1]$ for all $x\in \{0,1\}^k$.
\end{proof}

We will also need the following fact about multilinear polynomials.

\begin{fact}[Multilinear maximum principle]\label{fact:multilinear_max}
    Let $p\in \mathbb{R}[X_1,\dots,X_n]$. Suppose $p$ is  multilinear, then
    \begin{equation}
        \max_{x\in \{0,1\}^n} |p(x)|=\max_{\mu\in [0,1]^n} |p(\mu)|.
    \end{equation}
\end{fact}

\begin{proof}
It suffices to prove $\max_{\mu\in [0,1]^n} |p(\mu)| \leq \max_{x\in \{0,1\}^n} |p(x)|$ as the reverse inequality is clear. Fix $\mu\in [0,1]^n$. Let $\ber_{\mu}$ denote the distribution on $\{0,1\}^n$ where the $i$th bit is sampled to be $1$ with probability $\mu_i$ independently. Since $p$ is multilinear, we have $p(\mu)=\expect_{x\sim \ber_{\mu}}[p(x)]$. Therefore,
\begin{equation}
    |p(\mu)| = \bigl| \expect_{x\sim \ber_{\mu}}[p(x)] \bigr|\leq \expect_{x\sim \ber_{\mu}}[|p(x)|]\leq \max_{x\in \{0,1\}^n} |p(x)|.
\end{equation}
Since this holds for all $\mu\in [0,1]^n$, we obtain $\max_{\mu\in [0,1]^n} |p(\mu)|\leq \max_{x\in \{0,1\}^n} |p(x)|$, as required.
\end{proof}

\begin{theorem}\label{thm:fbs_leq_sdeg}
    For every nonconstant Boolean function $f$,
    \begin{equation}
        \fbsmindown(f)\leq \frac{\pi^2}{2}\sdeg(f)^2.
    \end{equation}
\end{theorem}

\begin{proof}
    Let $f\colon \{0,1\}^n\to \{0,1\}$.
    Let $p$ be a sign representation of $f$ with $\deg(p) = \sdeg(f)$. Let $z\in \{0,1\}^n$ be such that $\abs{p(z)}$ is maximized.
    
    Fix an arbitrary positive integer $k$. Let $q \in \mathbb{R}[X_1,\dots,X_k]$ be the polynomial from \cref{lem:approx_or}. Write $q^{(0)}$ for $q$ and $q^{(1)}$ for $1-q$. 
    
    Then define an $(nk)$-variate polynomial $r\in \mathbb{R}[X_{1,1},\dots,X_{1,k},\dots,X_{n,1},\dots,X_{n,k}]$ by
    \begin{equation}
    r \coloneqq p\bigl( \, q^{(z_1)}(X_{1,1},\dots,X_{1,k}),\dots,q^{(z_n)}(X_{n,1},\dots,X_{n,k}) \, \bigr),
    \end{equation}
    and define $g\colon \{0,1\}^{nk} \to \{0,1\}$ by $g(x)=1$ if and only if $r(x) < 0$. For convenience, write $\mathbf{0} \coloneqq 0^{nk}$.
    
    From these definitions, it is clear that $r(\mathbf{0}) = p(z)$ and that
    \begin{equation}\label{eq:r_sign_deg}
        \deg(r)\leq \deg(p)\cdot \deg(q)=\sdeg(f)\cdot \left\lceil \frac{\pi}{2}\sqrt{k}\right\rceil.
    \end{equation}
    
    Since $p$ is multilinear, and $q(x) \in [0,1]$ and $1-q(x)\in [0,1]$ for all $x\in \{0,1\}^k$, \cref{fact:multilinear_max} implies $\max_{v\in \{0,1\}^{nk}}\abs{r(v)}\leq \abs{p(z)}$ and equality holds at $v=\mathbf{0}$. Therefore, by the same argument used in the proof of \cref{lem:sdeg_lower}, we have
    \begin{equation}\label{eq:sdeg_bs}
        \bs_{\mathbf{0}}(g)\leq 2\deg(r)^2.
    \end{equation}
    
    We now show that \begin{equation}\label{eq:fbs_leq_bs}
        k\cdot\fbs_{z}(f)-2^n\leq \bs_{\boldzero}(g).
    \end{equation}
    
    Let $\mathcal{B}$ be the collection of sensitive blocks of $f$ at $z$. Let $\{w_B\}_{B\in \mathcal{B}}$ be an optimal solution to the linear program defining $\fbs_z(f)$. Define integer $W_B \coloneqq \lfloor k\cdot w_B\rfloor$ for each $B\in \mathcal{B}$. Then,
    \begin{equation}\label{eq:wb_lower}
        \sum_{B\in \mathcal{B}} W_B\geq \sum_{B\in \mathcal{B}}(k\cdot w_B-1)= k\cdot \fbs_z(f)-|\mathcal{B}|\geq k\cdot \fbs_z(f)-2^n.
    \end{equation}

    For each $B\in \calB$, we can associate one or more sensitive blocks of $g$ at $\mathbf{0}$ as follows. Write $B = \{i_1,\dots,i_m\}$. Then associate the block $B_{\lift} \coloneqq \{(i_1,j_1),\dots,(i_m,j_m)\}$ for any choice of $j_1,\dots,j_m\in [k]$. Using the definition of $q$, we see that $r(\mathbf{0}^{B_{\lift}}) = p(z^B)$. Since $p$ sign-represents $f$, and $r$ sign-represents $g$, the block $B_{\lift}$ is a sensitive block of $g$ at $\boldzero$. In fact, for each $B\in \calB$, we can associate $W_B$ sensitive blocks of $g$ at $\boldzero$, forming a collection $\lifts(B)$, such that all blocks in $\cup_{B\in \calB} \, \lifts(B)$ are pairwise disjoint. This is because, for all $i\in [n]$, we have
    \begin{equation}
        \sum_{B\ni i} W_B\leq \sum_{B\ni i}k\cdot w_B\leq k.
    \end{equation}
    Therefore,
    \begin{equation}\label{eq:wb_upper}
        \sum_{B\in \mathcal{B}} W_B\leq \bs_{\boldzero}(g).
    \end{equation}
    Combining \cref{eq:wb_lower,eq:wb_upper} gives \cref{eq:fbs_leq_bs} as claimed. By further combining \cref{eq:fbs_leq_bs} with \cref{eq:r_sign_deg,eq:sdeg_bs}, we obtain
    \begin{equation}
        k\cdot \fbs_z(f)-2^n\leq 2\Bigl(\sdeg(f)\cdot \left\lceil \frac{\pi}{2}\sqrt{k}\right\rceil\Bigr)^2.
    \end{equation}
    Since $k$ is arbitrary, we can divide the preceding equation by $k$ and take the $k\to \infty$ limit to obtain
    \begin{equation}
        \fbs_z(f)\leq \frac{\pi^2}{2}\sdeg(f)^2.
    \end{equation}
    Since $\fbs_{\min}(f)\leq \fbs_z(f)$, we have
    \begin{equation}
        \fbs_{\min}(f)\leq \frac{\pi^2}{2}\sdeg(f)^2.
    \end{equation}
    The theorem follows upon noting that the preceding equation holds for all Boolean functions $f$ and that $\deg_{\pm}(f|_\rho) \leq \deg_{\pm}(f)$ for all restrictions $\rho$.
\end{proof}

\subsection{Putting everything together}

By combining \cref{eq:rcmin_fbs}, \cref{thm:rcmin}, and \cref{thm:fbs_leq_sdeg}, we obtain
\begin{theorem}\label{thm:rdeg_cubed}
    For every $f\colon \{0,1\}^n \to \{0,1\}$, 
    \begin{equation}\label{eq:rdeg_cubed}
        \D(f) \leq O\bigl(\rdeg(f) \, \sdeg(f)^2 \, \log{n}\bigr) \leq O\bigl(\rdeg(f)^3 \, \log n\bigr).
    \end{equation}
\end{theorem}

\begin{remark}
    The first inequality in \cref{eq:rdeg_cubed} can be tight up to a $\log n$ factor as witnessed by the majority function. It may also be difficult to significantly improve $\D(f) \leq O(\rdeg(f)^3 \log n)$ since the best-known upper bound on $\D(f)$ by even $\deg(f)$ is $\D(f) \leq O(\deg(f)^3)$ \cite{ddeg_midrijanis_2004}.
\end{remark}

\cref{thm:rdeg_cubed} is asymptotically stronger than \cref{thm:D_upper} because of the next theorem. To prove the theorem, we need the notion of influence. For $f\colon \{0,1\}^n \to \{0,1\}$ and $i\in [n]$, the $i$th \emph{influence} of $f$ is defined by $\Inf_i[f] \coloneqq \Pr[f(x) \neq f(x^i)]$, where $x^i$ is $x$ with the $i$th bit flipped, and the probability is over uniformly random $x\in \{0,1\}^n$. The \emph{total influence} of $f$ is defined by $\Inf[f]\coloneqq\sum_{i=1}^n \Inf_i[f]$.
    
\begin{theorem}\label{thm:logn}
    For every $f\colon \{0,1\}^n \to \{0,1\}$ that depends on all $n$ variables, $\rdeg(f) \geq \Omega(\log n)$.
\end{theorem}

\begin{proof}
    Let $p/q\in \mathbb{R}(X_1,\dots,X_n)$ be a rational representation of $f$ such that $\max(\deg(p),\deg(q))=\rdeg(f)$. Fix an arbitrary $i\in [n]$ and define $g \colon \{0,1\}^n \to \{-1,0,1\}$ by 
    \begin{equation}
        g(x) \coloneqq f(x)-f(x^i)=\frac{p(x)}{q(x)}-\frac{p(x^i)}{q(x^i)}=\frac{p(x)q(x^i)-p(x^i)q(x)}{q(x)q(x^i)},
    \end{equation}
    where $x^i$ denotes $x$ with the $i$th bit flipped.

    From the numerator of $g$, we obtain $r\in \mathbb{R}[X_1,\dots,X_n]$ such that $r(x) = p(x)q(x^i)-p(x^i)q(x)$ for all $x\in \{0,1\}^n$ and $\deg(r)\leq \deg(p)+\deg(q)\leq 2\rdeg(f)$. Observe that $g(x)\neq 0$ if and only if $r(x)\neq 0$. Moreover, observe that $g$ is not constantly zero since $f$ depends on variable $i$ by assumption. Therefore, \cref{thm:ns94} gives
    \begin{equation}\label{eq:infi_bound}
        \Inf_i[f]=\Pr[g(x)\neq 0]=\Pr[r(x)\neq 0]\geq 2^{-\deg(r)}\geq 2^{-2\rdeg(f)}.
    \end{equation}
    
    We now sum \cref{eq:infi_bound} over all $i\in [n]$. Combined with \cref{thm:rdeg_cubed} and the fact that $\Inf[f]\leq \deg(f)$, this yields
    \begin{equation}\label{eq:inf_bounds}
        \frac{n}{2^{2\rdeg(f)}}\leq \sum_{i=1}^n \Inf_i[f]=\Inf[f]\leq \deg(f)\leq O(\rdeg(f)^3 \log n).
    \end{equation}
    Rearranging gives $\rdeg(f)\geq \Omega(\log n)$ as required.
\end{proof}

\begin{remark}
    The bound in \cref{thm:logn} is tight for the address function \cite{survey_buhrman_2002}. \cref{eq:inf_bounds} in the proof shows that a Boolean function $f$ with rational degree $d$ can depend on at most $O(d^4 2^{2d})$ variables. It may be possible to improve this bound along the lines of \cite{ninputs_chiarelli_2020,ninputs_wellens_2022}.
\end{remark}

We now present some direct corollaries of \cref{thm:rdeg_cubed} that we find most interesting. The first shows that the $\log n$ factor in \cref{eq:rdeg_cubed} can be replaced by $\log \rdeg(f)$ if the left-hand side is relaxed to $\deg(f)$, $\s(f)$, or $\bs(f)$. Here $\s(f)$ denotes the sensitivity of $f$ which is defined as follows. For $f\colon \{0,1\}^n \to \{0,1\}$ and $x\in \{0,1\}^n$, write $\s_x(f)$ for the size of the set $\{i\in [n]\colon f(x) \neq f(x^i)\}$, where $x^i$ denotes $x$ with the $i$th bit flipped, then $\s(f) \coloneqq \max_{x\in \{0,1\}^n}\s_x(f)$.

\begin{corollary}\label{cor:deg_rdegcubed}
    For every Boolean function $f$,
    \begin{alignat}{2}
        \deg(f) ~\leq&~ O\bigl(\rdeg(f) \, \sdeg(f)^2 \, \log \rdeg(f)\bigr) ~\leq&~ \tO\bigl(\rdeg(f)^3\bigr),\label{eq:deg_rdegcubed}
        \\
        \s(f) \leq \bs(f) ~\leq&~ O\bigl(\rdeg(f) \, \sdeg(f)^2 \, \log \rdeg(f)\bigr) ~\leq&~ \tO\bigl(\rdeg(f)^3\bigr).\label{eq:bs_rdegcubed}
    \end{alignat}
\end{corollary}

\begin{proof}
    Consider \cref{eq:deg_rdegcubed} first. As $\rdeg(f)\sdeg(f)^2\log \rdeg(f) \leq \tO(\rdeg(f)^3)$, it suffices to prove the first inequality. Let $p$ be the polynomial representation of $f$ and let $M$ be a monomial of $p$ with $\abs{M} = \deg(f)$. Let $f|_M$ denote an arbitrary restriction of $f$ to $M$. Note that $f|_M$ is defined on $\deg(f)$ variables. Applying \cref{thm:rdeg_cubed} to $f|_M$ gives
    \begin{equation}
        \deg(f)=\deg(f|_M)\leq \D(f|_M)\leq O\bigl(\rdeg(f|_M)\cdot \sdeg(f|_M)^2 \cdot \log \deg(f)\bigr).
    \end{equation}
    \cref{eq:deg_rdegcubed} follows from the facts that $\rdeg(f|_M)\leq \rdeg(f)$ and $\sdeg(f|_M) \leq \sdeg(f)$.\footnote{If $a/\log a \leq O(b)$, there is a constant $C$ such that $a\le Cb\log a$ for all large $a$. Taking logarithms gives $\log a\le \log b+\log\log a+O(1)$, hence $\log a \leq O(\log b)$; substituting back gives $a\le Cb\log a \leq O(b\log b)$.}

    Now consider \cref{eq:bs_rdegcubed}. Let $x\in \{0,1\}^n$ be such that $\bs(f) = \bs_x(f)$ and let $\{B_1,\dots,B_{\bs_x(f)}\}$ be a matching set of disjoint \emph{minimal} sensitive blocks of $f$ at $x$, where $|B_i|\leq \s(f)\leq \bs(f)$ by minimality. Take $S\coloneqq \bigsqcup_{i=1}^{\bs_x(f)} B_i$. Let $f|_S$ denote the restriction of $f$ to $S$ by fixing variables outside $S$ according to $x$. Since $\abs{S} = \sum_{i=1}^{\bs_x(f)}|B_i| \leq \sum_{i=1}^{\bs_x(f)} \bs(f)=\bs(f)^2$, $f|_S$ is defined on at most $\bs(f)^2$ variables. Applying \cref{thm:rdeg_cubed} to $f|_S$ gives
    \begin{equation}
       \bs(f)=\bs(f|_S)\leq \D(f|_S)\leq O\bigl(\rdeg(f|_S)\cdot \sdeg(f|_S)^2 \cdot \log (\bs(f)^2)\bigr).
    \end{equation}
    \cref{eq:bs_rdegcubed} follows from the facts that $\rdeg(f|_S)\leq \rdeg(f)$ and $\sdeg(f|_S) \leq \sdeg(f)$.
\end{proof}

\begin{remark}
    One may attempt a similar proof strategy to show $\D(f) \leq \tO(\rdeg(f)^3)$. For this to work, it suffices to show that $\D(f)$ satisfies the following \emph{hardness-condensation} property: for every Boolean function $f$, there exists a restriction $f'$ of $f$ to $\mathrm{poly}(\D(f))$ variables such that $\D(f) \leq O(\D(f'))$. Hardness condensation has been studied previously in \cite{condensation_goos_2024} for example.
\end{remark}

By adapting the proof of \cref{thm:rdeg_cubed}, we can obtain the following corollary.

\begin{corollary}\label{cor:final_sdeg}
    For every Boolean function $f$, 
    \begin{equation}\label{eq:final_fdeg_sdeg}
        \deg(f)\leq \tO\bigl(\deg_{\mathbb{F}_2}(f) \, \sdeg(f)^2\bigr),
    \end{equation}
    where $\deg_{\mathbb{F}_2}(f)$ is the minimum degree of $p\in \mathbb{F}_2[X_1,\dots,X_n]$ such that $\forall x \in \{0,1\}^n, \, p(x) = f(x)$.
\end{corollary}

\begin{proof}[Proof sketch] \cref{thm:rcmin} still holds with $\rdeg(f)$ replaced by $\deg_{\mathbb{F}_2}(f)$, as \cref{lem:nisan-smolensky} holds when $p$ is a polynomial representation of $f$ over $\mathbb{F}_2$.
\end{proof}

\begin{remark}
    We find \cref{cor:final_sdeg} surprising because neither $\deg_{\mathbb{F}_2}(f)$ nor $\sdeg(f)$ are polynomially related to $\deg(f)$, as witnessed by the parity and majority functions, respectively. This is reminiscent of how neither $\C_0(f)$ nor $\C_1(f)$ are polynomially related to $\deg(f)$ yet $\deg(f) \leq \C_0(f) \, \C_1(f)$, where $\C_b(f) \coloneqq \max_{x\in f^{-1}(b)}\C_x(f)$ denotes the $b$-certificate complexity. But unlike $\C_0(f)$ and $\C_1(f)$, both $\deg_{\mathbb{F}_2}(f)$ and $\sdeg(f)$ stay invariant under negating $f$.
\end{remark}

The next two corollaries concern $\ndeg(f)$ and $\ndeg(\neg f)$.

\begin{corollary}\label{cor:final_ndeg}
    For every $f\colon \{0,1\}^n \to \{0,1\}$, 
    \begin{equation}\label{eq:final_ndeg}
        \D(f)\leq O\bigl(\ndeg(f)^{1.5} \, \ndeg(\neg f)^{1.5} \, \log n\bigr) \quad \text{and} \quad  \D(f)\leq O\bigl(\ndeg(f)^2 \, \ndeg(\neg f)^2 \bigr).
    \end{equation}
\end{corollary}

\begin{proof} The first inequality follows from \cref{thm:rdeg_cubed} since $\sdeg(f) \leq O(\min(\ndeg(f),\ndeg(\neg f)))$ and $\rdeg(f) = \max(\ndeg(f),\ndeg(\neg f))$. The second inequality additionally uses \cref{thm:logn}.
\end{proof}

\begin{corollary}\label{cor:final_deg_ndeg}
    For every Boolean function $f$,
    \begin{equation}\label{eq:final_deg_ndeg}
        \deg(f)\leq O\bigl(\ndeg(f)^{1.5} \, \ndeg(\neg f)^{1.5} \, \log\rdeg(f)\bigr) \leq \tO\bigl(\ndeg(f)^{1.5} \, \ndeg(\neg f)^{1.5}\bigr).
    \end{equation}
\end{corollary}

\begin{proof} The corollary follows from \cref{cor:final_ndeg} by the same argument as in the proof of \cref{cor:deg_rdegcubed}.
\end{proof}

\section{Implications and open problems}
\label{sec:implications}

\subsection{Effective Hypercube Nullstellensatz}\label{sec:nullstellensatz}

The main result of this work can be framed as an effective Hypercube Nullstellensatz.

\begin{theorem}[Effective Hypercube Nullstellensatz]\label{null_cube}
    Let $g_1, g_2 \in \mathbb{R}[X_1,\dots,X_n]$. Suppose $g_1$ and $g_2$ do not share any common zeros on the hypercube $\{0,1\}^n$. Further suppose $g_1(x) \cdot g_2(x) = 0$ for all $x\in\{0,1\}^n$. Then there exist $h_1, h_2\in \mathbb{R}[X_1, \dots, X_n]$ such that
    \begin{equation}\label{eq:bezout}
        h_1(x) g_1(x) + h_2(x) g_2(x) = 1 \quad \text{for all $x\in \{0,1\}^n$},
    \end{equation}
    and
    \begin{equation}
        \max\bigl(\deg(\overline{h_1 g_1}), \deg(\overline{h_2 g_2})\bigr) \leq \tO\bigl(\deg(g_1)^{1.5} \deg(g_2)^{1.5}\bigr),
    \end{equation}
    where the overline denotes multilinearization using the relations $X_1^2=X_1,\dots,X_n^2=X_n$.
\end{theorem}

\begin{proof}
    Construct polynomials $h_1, h_2$ satisfying \cref{eq:bezout} by interpolation. We proceed to bound $\max(\deg(\overline{h_1 g_1}), \deg(\overline{h_2 g_2}))$. Define $f\colon \{0,1\}^n \to \{0,1\}$ by $f(x) = 0$ if and only if $g_1(x) = 0$. Then by the hypotheses of the theorem, $g_1, g_2$ are nondeterministic representations of $f$ and $\neg f$, respectively. For every $x\in \{0,1\}^n$, we have $h_1(x) g_1(x) = f(x)$ and $h_2(x) g_2(x) = \neg f(x)$. Therefore by uniqueness of multilinear representation, we deduce $\max(\deg(\overline{h_1 g_1}),\deg(\overline{h_2 g_2}))=\max(\deg(f), \deg(\neg f))$, and the theorem follows from \cref{cor:final_deg_ndeg}.
\end{proof}

In view of existing Nullstellensatz results, in particular \cite{effective_jelonek_2005}, we conjecture that a natural generalization of \cref{null_cube} to any number of polynomials holds.

\begin{conjecture}
    For all integers $m\geq 2$, the following holds. Let $0\neq g_1, \dots, g_m\in \mathbb{R}[X_1,\dots,X_n]$. Suppose $g_1, \dots, g_m$ do not share any common zeros on the hypercube $\{0,1\}^n$. Further suppose $g_1(x) \cdots g_m(x) = 0$ for all $x\in \{0,1\}^n$. Then there exist $h_1,\dots, h_m\in \mathbb{R}[X_1,\dots,X_n]$ such that
    \begin{equation}
        h_1(x) g_1(x)+\cdots + h_m(x) g_m(x)=1 \quad \text{for all $x\in \{0,1\}^n$},
    \end{equation}
    and
    \begin{equation}
        \max_{i\in [m]}(\deg(\overline{h_i g_i}))\leq \mathrm{poly}(\deg(g_1),\dots,\deg(g_m)).
    \end{equation}
\end{conjecture}

Again in view of existing Nullstellensatz results, a possible strengthening of the conjecture would have the condition ``$g_1(x) \cdots g_m(x)=0$ for all $x\in \{0,1\}^n$'' removed. Such a conjecture would mean that an effective Nullstellensatz holds for all subsets of the hypercube. However, this conjecture is false, even when $m=2$, as we demonstrate below using a specific subset.\footnote{In computer science language, \cref{fact:counter_example} shows that rational degree could be much smaller than degree for partial Boolean functions $f$, even when the (rational) polynomial representation is not required to be bounded outside the domain of $f$. Such separation cannot be shown by the ``Boolean Imbalance'' function of \cite{rdeg_iyer_2025}, or others like it.}

\begin{fact}\label{fact:counter_example}
    There exist $g_1,g_2 \in \mathbb{R}[X_1,\dots,X_n,Y_1,\dots,Y_n]$ each of degree $1$ that do not share any common zeros on $\{0,1\}^{2n}$ such that: if $h_1, h_2 \in \mathbb{R}[X_1,\dots,X_n,Y_1,\dots,Y_n]$ satisfy 
    \begin{equation}\label{eq:bezout_copy}
        h_1(x) g_1(x) + h_2(x) g_2(x) = 1 \quad \text{for all $x\in \{0,1\}^{2n}$},
    \end{equation}
    then 
    \begin{equation}
        \max\bigl(\deg(\overline{h_1 g_1}), \deg(\overline{h_2 g_2})\bigr) \geq n.
    \end{equation}
\end{fact}

\begin{proof} We give an explicit construction.
    Let $g_1, g_2\in \mathbb{R}[X_1,\dots,X_n,Y_1,\dots,Y_n]$ be defined by
    \begin{equation}
        g_1\coloneqq X_1+\cdots+X_n \quad \text{and} \quad g_2 \coloneqq X_1+\cdots+X_n+Y_1+\cdots+Y_n - (n+1).
    \end{equation}
    Clearly, $\deg(g_1) = \deg(g_2) = 1$, and $g_1, g_2$ do not share any common zeros on $\{0,1\}^{2n}$. 
    
    Let disjoint sets $D_0,D_1 \subseteq \{0,1\}^n \times \{0,1\}^n$ be defined by
    \begin{align}
         D_0~\coloneqq&~ \{(x,y)\in \{0,1\}^n \times \{0,1\}^n\colon g_1(x,y)=0\},
         \\
         D_1~\coloneqq&~\{(x,y)\in \{0,1\}^n \times \{0,1\}^n\colon g_2(x,y)=0\}.
    \end{align}
    Let $D \coloneqq D_0 \sqcup D_1$ and define $f\colon D\to \{0,1\}$ by $f(x,y)=0$ if and only if $(x,y)\in D_0$. 
    
    We claim that if $p\in \mathbb{R}[X_1,\dots,X_n,Y_1,\dots,Y_n]$ satisfies $p(x,y)=f(x,y)$ for all $(x,y) \in D$, then $\deg(p)\geq n$. This directly implies that, if $h_1, h_2$ satisfy \cref{eq:bezout_copy}, then $\deg(\overline{h_1  g_1})\geq n$.

    We prove the claim using Minsky-Papert symmetrization (\cref{thm:symmetrization}). By symmetrizing $p$, first with respect to the $X_i$s and then the $Y_i$s, we obtain $P\in \mathbb{R}[S,T]$ such that $\deg(P) \leq \deg(p)$ and 
    \begin{align}
        P(0,t)~&=~0, && \text{for all $t\in \{0,1,\dots, n\}$};\label{eq:root_overflow_1}
        \\
        P(s,n+1-s)~&=~1, && \text{for all $s\in \{1,\dots, n\}$}.\label{eq:root_overflow_2}
    \end{align}
    We then perform case analysis based on the degree of the univariate polynomial $P(0,T)\in \mathbb{R}[T]$:
    \begin{enumerate}
        \item Case $\deg(P(0,T)) \geq n$. Then  $\deg(p) \geq \deg(P) \geq \deg(P(0,T)) \geq n$ and the claim holds.
        \item Case $\deg(P(0,T)) < n$. Since $P(0,T)$ has at least $n+1$ roots by \cref{eq:root_overflow_1} and degree less than $n$, it must be the zero polynomial. Therefore, the polynomial $P(S,n+1-S)-1$ is not identically zero, since $P(0,n+1)-1=-1$. Moreover, by \cref{eq:root_overflow_2}, we see that $P(S,n+1-S)-1$ has at least $n$ roots.  Therefore, $\deg(P(S,n+1-S)-1)\geq n$. This shows that
        \begin{equation}
            n\leq \deg(P(S,n+1-S)-1)\leq \deg(P)\leq \deg(p),
        \end{equation}
        and the claim holds.
    \end{enumerate}
    Since the claim holds in either case, the fact follows.
\end{proof}

\subsection{Gotsman-Linial conjecture}

The long-standing Gotsman-Linial conjecture~\cite{spectral_gotsman_1994} posits that for every $f\colon \{0,1\}^n \to \{0,1\}$,  $\Inf[f] \leq O(\sqrt{n} \sdeg(f))$. We conjecture the following, which is weaker since $\sdeg(f)/2 \leq \ndeg(f)$.
\begin{conjecture}\label{conj:gotsman_linial}
    For every $f\colon \{0,1\}^n \to \{0,1\}$, $\Inf[f]\leq O(\sqrt{n} \ndeg(f))$.
\end{conjecture}

The motivation for this conjecture is rooted in the origins of this work, namely the observation that existing results on the Gotsman-Linial conjecture \cite{influence_diakonikolas_2014,bounding_harsha_2014,correct_kane_2014} together with an algebraic argument imply $\rdeg(f) \geq \Omega(\sqrt{\log n})$ for every Boolean function $f$ that depends on $n$ variables.\footnote{For example, \cite[Theorem 1.1]{influence_diakonikolas_2014} shows $\Inf[f] \leq 2^{O(\sdeg(f))} \cdot \log n \cdot n^{1-1/(4\sdeg(f)+2)}$, which implies that the same inequality holds with $\sdeg(f)$ replaced by $2\rdeg(f)$. Combining this with  $n/2^{2\rdeg(f)} \leq \Inf[f]$ --- see proof of \cref{thm:logn} --- yields $\rdeg(f) \geq \Omega(\sqrt{\log n})$. (Also note that \cite{influence_diakonikolas_2014} uses the notation ``$\AS$'' for $\Inf$.)} Proving that $\Inf[f] \leq O(\sqrt{n} \ndeg(f))$ would yield $\rdeg(f) \geq \Omega(\log n)$. We deferred pursuing \cref{conj:gotsman_linial} after obtaining \cref{thm:D_upper}, since that theorem yields $\rdeg(f) \geq \Omega(\log n)$ more directly, as the proof of \cref{thm:logn} shows.

\subsection{Approximate nondeterministic degree}

For $f\colon \{0,1\}^n \to \{0,1\}$ and $\epsilon\in [0,1)$, define the $\epsilon$-approximate nondeterministic degree of $f$, denoted $\ndeg_{\epsilon}(f)$, to be the minimum degree of a real polynomial $p$ such that, for all $x \in \{0,1\}^n$,
\begin{equation}
\begin{cases}
    \abs{p(x)} \leq \epsilon & \text{if } f(x) = 0, \\[1pt]
    \abs{p(x)} \geq 1 & \text{if } f(x) = 1.
\end{cases}
\end{equation}
It is easy to see that $\ndeg_0(f) = \ndeg(f)$ for all $f$. The value of $\ndeg_\epsilon$ also coincides (up to constants) with ``$\deg_\epsilon^+$'' of \cite{breaking_sherstov_2018} and ``$\widetilde{\odeg}_\epsilon$'' of \cite{apxdeg_bun_2022}.

For constant $\epsilon$, it is easy to adapt the analytic side of our proofs, say \cref{lem:sdeg_lower}, to obtain 
\begin{theorem}\label{thm:apx_version}
For every Boolean function $f$, and constant $\epsilon\in [0,1)$, 
\begin{equation}
    \D(f) \leq O(\ndeg_\epsilon(f)^2 \ndeg(\neg f)^2) \quad \text{and} \quad \D(f) \leq O(\ndeg(f)^2 \ndeg_\epsilon(\neg f)^2).
\end{equation}
\end{theorem}
If we could similarly adapt the combinatorial side of our proofs, say \cref{lem:hitset_upper}, then we would prove the following conjecture.
\begin{conjecture}
    For every Boolean function $f$, and constant $\epsilon\in [0,1)$,
    \begin{equation}
        \D(f)\leq O(\ndeg_\epsilon(f)^2\ndeg_\epsilon(\neg f)^2).
    \end{equation}
\end{conjecture}
Some progress towards this conjecture for special classes of functions is reported in \cite{apxndeg_pednekar_2026}. Towards this conjecture, we prove the following theorem in \cref{app:D_apxndeg}.
\begin{restatable}{theorem}{apxndeg}\label{thm:apxndeg}
    For every Boolean function $f$, and constant $\epsilon\in [0,1)$,
    \begin{equation}
        \D(f)\leq O(\ndeg_\epsilon(f)^3\ndeg_\epsilon(\neg f)^3).
    \end{equation}
\end{restatable}

\subsection{Improving polynomial relations}

Since degree is polynomially related to almost all other Boolean complexity measures, \cref{cor:deg_rdegcubed} immediately yields polynomial relations between rational degree and those measures too. It would be interesting to see the extent to which these polynomial relations could be tightened. More specifically, the main results of this work turn all question marks in \cite[Table 1]{rdeg_iyer_2025} into either the number $3$ or $4$ (depending on whether a $\log n$ factor can be removed), but can we obtain matching numbers? In particular, we conjecture that \cref{thm:rdeg_cubed} is optimal up to a $\log n$ factor.
\begin{conjecture}\label{conjecture:quartic_separation}
    There exists a family of Boolean functions $f$ such that $\D(f) \geq \Omega(\rdeg(f)^3)$.
\end{conjecture}

The currently best-known separation is quadratic. It can be witnessed by at least two different functions: the balanced $\ANDOR$ tree,\footnote{The balanced $\ANDOR$ tree on $m^2$ variables also simultaneously separates sign degree, rational degree, and degree: $\sdeg=O(\sqrt{m\log m})$  \cite{apxdeg_bun_2022} (also see \cref{rem:sdeg_lower}), $\rdeg=m$  \cite{rdeg_iyer_2025}, $\deg=m^2$.} or the ``pointer function'' that quadratically separates $\D$ from exact quantum query complexity \cite{pointer_ambainis_2017}. 

We note that \cite[Section 4]{nondeterministic_dewolf_2003} conjectured\footnote{\cite[Section 4]{nondeterministic_dewolf_2003} made this conjecture in the form $\D(f) \leq O(\NQ(f) \, \NQ(\neg f))$, but the same paper also showed that $\ndeg(f) = \NQ(f)$. For comparison, \cref{cor:final_ndeg} implies $\D(f) \leq O(\ndeg(f)^2 \ndeg(\neg f)^2)$.} $\D(f) \leq O(\ndeg(f) \, \ndeg(\neg f))$ for every Boolean function $f$.  If true, this would falsify \cref{conjecture:quartic_separation}. However, at the time \cite{nondeterministic_dewolf_2003} was published, contrived Boolean functions like those constructed in \cite{cheatsheets_aaronson_2016,pointer_ambainis_2017} were unknown. More
recently (in 2026), Ronald de Wolf informed us of a weaker conjecture:
\begin{equation}\label{eq:ronald}
     Q_E(f) \overset{?}{\leq} O(\ndeg(f) \, \ndeg(\neg f)),
\end{equation}
where $Q_E$ denotes the exact quantum query complexity of $f$. If \cref{eq:ronald} held, it would be a quantum counterpart to the classical fact $\D(f) \leq \C_0(f) \, \C_1(f)$.

\section*{Acknowledgments}

We thank Lance Fortnow, Joel Friedman, Zbigniew Jelonek, and Ronald de Wolf for helpful discussions, comments, and suggestions. We acknowledge the use of ChatGPT, Claude, and Gemini to search the literature and brainstorm proof strategies for \cref{sec:improved_lower_bound}. 

We thank anonymous FOCS 2026 reviewers for their comments. Subsequent to our FOCS 2026 submission, we added \cref{thm:apxndeg}, which was initially proven by GPT 5.6-Sol given the source files for an earlier version of this paper. We simplified and modularized the proof.

\appendix
\crefalias{section}{appendix}

\section{Rational degree and quantum postselection}
\label{app:rdeg_postqe}

In this appendix, we show that rational degree exactly equals the zero-error postselected quantum query complexity. We will need the following technical lemma from \cite{rdeg_iyer_2025}.
\begin{lemma}[{\cite[Lemma 26]{rdeg_iyer_2025}}]\label{lem:avoidance_lemma}
    Let $N$ be a positive integer. Let $D$ be a finite set. Let $a_1,\dots, a_N \colon D \to \mathbb{R}$. Suppose that for all $x \in D$, there exists $i\in [N]$ such that $a_i(x) \neq 0$. Then there exist $c_1,\dots, c_N >0 $ such that, for all $x\in D$, $(c_1 a_1 + \dots + c_N a_N)(x) \coloneqq c_1 a_1(x) + \dots + c_N a_N(x) \neq 0$.
\end{lemma}
Since the proof is short, we reproduce it below for completeness.
\begin{proof}[Proof of \cref{lem:avoidance_lemma}]
    Since $D$ is finite, for all $i\in [N]$, there exists $0< b_i < B_i$ such that $a_i(D) \setminus \{0\}$ is a subset of $(-B_i, -b_i) \cup (b_i, B_i)$. Now define $c_1 = 1$ and for $i = 2,\dots, N$, define $c_i > 0$ by 
    \begin{equation}
        c_i b_i \coloneqq 1 + \sum_{j=1}^{i-1}c_j B_j> \sum_{j=1}^{i-1}c_j B_j.
    \end{equation}
    It is easy to verify that
    $(c_1 a_1 + \dots + c_N a_N)(x) \neq 0$ for all $x\in D$, as required.
\end{proof}

We will use \cref{lem:avoidance_lemma} to show\footnote{\cite{rdeg_iyer_2025} uses \cref{lem:avoidance_lemma} for two other purposes: proving an AND-composition lemma for nondeterministic degree, and proving an upper bound on the rational degree of their ``Middle Third'' function.}

\rdegpostqe*

Our proof assumes familiarity with the notation and definitions of \cite{postqe_mahadev_2015}.

\begin{proof}
    \cite[Theorem 2]{postqe_mahadev_2015} gives $\rdeg(f) \geq \PostQ_0(f)$ so it suffices to prove $\rdeg(f) \leq \PostQ_0(f)$. 
    
    Write $Q\coloneqq \PostQ_0(f)$ for convenience. Suppose there exists a $Q$-query postselected quantum query algorithm $\calA$ that computes $f$ exactly with zero error. Then, using the polynomial method \cite{polynomial_beals_2001} following \cite[Proof of Theorem 1]{postqe_mahadev_2015}, we deduce that there exists an integer $m \geq 2$ and complex multilinear polynomials $\alpha_s \in \mathbb{C}[X_1,\dots,X_n]$ for all $s\in \{0,1\}^m$ such that:
    \begin{enumerate}
        \item For all $s \in \{0,1\}^m$, the degree of $\alpha_s$ is at most $Q$;
        \item For all inputs $x\in \{0,1\}^n$, when $\calA$ is run on $x$, its state just before postselection is
        \begin{equation}
            \ket{\psi(x)} \coloneqq \sum_{s\in \{0,1\}^m}\alpha_s(x)\ket{s}; 
        \end{equation}
        \item For all inputs $x\in \{0,1\}^n$, when $\calA$ is run on $x$, its probability of outputting $1$ is
        \begin{equation}\label{eq:r_eq_f}
            r(x) \coloneqq \frac{\sum_{s \in \{0,1\}^m \colon s_1 = 1, s_2 = 1} \abs{\alpha_s(x)}^2}{\sum_{s \in \{0,1\}^m \colon s_1 = 1, s_2 = 1} \abs{\alpha_s(x)}^2+\sum_{t \in \{0,1\}^m \colon t_1 = 1, t_2 = 0} \abs{\alpha_t(x)}^2} = f(x).
        \end{equation}
        (In particular, the denominator of $r(x)$ is strictly positive.)
    \end{enumerate}
    For convenience of notation, we will henceforth write
    \begin{equation}
        S \coloneqq \{s \in \{0,1\}^m \colon s_1 = 1, s_2 = 1\} \quad \text{and} \quad T \coloneqq \{t \in \{0,1\}^m \colon t_1 = 1, t_2 = 0\}.
    \end{equation}
    
    Then, \cref{eq:r_eq_f} can be written as
    \begin{equation}\label{eq:r_eq_f_simp}
         r(x) \coloneqq \frac{\sum_{s \in S} \abs{\alpha_s(x)}^2}{\sum_{s \in S} \abs{\alpha_s(x)}^2+\sum_{t \in T} \abs{\alpha_t(x)}^2} = f(x),
    \end{equation}
    which implies:
    \begin{itemize}
        \item For all $x\in f^{-1}(0)$, we have $\sum_{s\in S} \abs{\alpha_s(x)}^2=0$ and $\sum_{t\in T} \abs{\alpha_t(x)}^2\neq 0$. Therefore, $\alpha_s(x) = 0$ for all $s\in S$, and $\alpha_t(x) \neq 0$ for some $t\in T$.
        \item For all $x\in f^{-1}(1)$, we have $\sum_{s\in S} \abs{\alpha_s(x)}^2\neq 0$ and $\sum_{t\in T} \abs{\alpha_t(x)}^2 = 0$. Therefore, $\alpha_s(x) \neq 0$ for some $s\in S$, and $\alpha_t(x) = 0$ for all $t\in T$. 
    \end{itemize}

    For all $s\in \{0,1\}^m$, decompose $\alpha_s \in \mathbb{C}[X_1,\dots,X_n]$ into its real and imaginary parts:
    \begin{equation}
        \alpha_s = a_{s,0} + i a_{s,1},
    \end{equation}
    where $a_{s,0}, a_{s,1} \in \mathbb{R}[X_1,\dots,X_n]$ each have degree at most $Q$.

    Since a complex number is zero if and only if its real and imaginary parts are both zero, we see:
    \begin{itemize}
        \item For all $x\in f^{-1}(0)$, we have  $a_s(x) = 0$ for all $s\in S \times \{0,1\}$, and $a_t(x) \neq 0$ for some $t\in T \times \{0,1\}$.
        \item For all $x\in f^{-1}(1)$, we have $a_s(x) \neq 0$ for some $s\in S\times \{0,1\}$, and $a_t(x) = 0$ for all $t\in T\times \{0,1\}$. 
    \end{itemize}
    
    Then, \cref{lem:avoidance_lemma} gives $c_t>0$ for all $t\in T \times \{0,1\}$ and $c_s>0$ for all $s\in S\times \{0,1\}$ such that:
    \begin{itemize}
        \item For all $x\in f^{-1}(0)$, we have $\sum_{t \in T\times \{0,1\}} c_t \cdot a_t(x) \neq 0$.
        \item For all $x\in f^{-1}(1)$, we have $\sum_{s \in S\times \{0,1\}} c_s \cdot a_s(x) \neq 0$.
    \end{itemize}

    Therefore, it is clear that $R \in \mathbb{R}(X_1,\dots,X_n)$ defined by
    \begin{equation}
        R \coloneqq \frac{\sum_{s\in S\times \{0,1\}} c_s \cdot a_s}{\sum_{s\in S \times \{0,1\}} c_s \cdot a_s+\sum_{t\in T\times \{0,1\}} c_t \cdot a_t}
    \end{equation}
    is a rational representation of $f$, and the degrees of $R$'s numerator and denominator are each at most $Q$. Therefore, $\rdeg(f) \leq Q = \PostQ_0(f)$, as required.
\end{proof}

\section{Proof of \texorpdfstring{$\D(f) \leq O(\ndeg_\epsilon(f)^3 \ndeg_\epsilon(\neg f)^3)$}{upper bound on D by approximate nondeterministic degrees}}
\label{app:D_apxndeg}

The proof employs the projection-$\AND$ complexity of a Boolean function. 
\begin{definition}[Projection-$\AND$ complexity] The projection-$\AND$ complexity of $f\colon \{0,1\}^n \to \{0,1\}$, denoted $\PAND(f)$, is defined to be the largest integer $k \geq 0$ such that there exists $x \in \{0,1\}^n$ with $f(x) = 1$ and disjoint blocks $B_1,\dots, B_k \subseteq [n]$ such that $f(x^{\cup_{i\in I}B_i}) = 0$ for every nonempty $I\subseteq[k]$. (Recall that $x^B$ denotes $x$ with all bits in $B$ flipped.)
\end{definition}

\begin{remark}Equivalently, $\PAND(f)$ is the largest integer $k\geq 0$ such that $\AND_k$ (equivalently $\NOR_k$) is a projection of $f$, where the meaning of projection is as in, e.g., \cite{projection_skyum_85}.
\end{remark}

For a Boolean function $f$, write $\C_{0,\min}(f) \coloneqq \min_{x\in f^{-1}(0)}\C_x(f)$, and write  $\C_{0,\min}^{\downarrow}(f)$ for the downward closure of $\C_{0,\min}(f)$ according to \cref{def:downward_closure}.

\begin{lemma}\label{lem:pnor}
    For every Boolean function $f$,
    \begin{equation}
        \D(f) \leq \C_{0,\min}^{\downarrow}(f) \cdot \PAND(f).
    \end{equation}
\end{lemma}
\begin{proof}
    Consider the deterministic query algorithm that iteratively queries the input $x$ at the positions holding a minimal-sized $0$-certificate, initially of $f$, then of its restrictions. If the function becomes constant after an iteration, the algorithm terminates and outputs that constant. It is clear that the algorithm must terminate after a finite number of iterations, say $k+1$. Write $B_i' \subseteq [n]$ for the positions queried at iteration $i$. Note that $B_1',\dots, B_{k+1}'$ are disjoint. Since there are $k+1$ iterations, there exists $y$ that agrees with $x$ on $B_1'\cup \dots \cup B_k'$ such that $f(y) = 1$.

    For each $i\in [k+1]$, let $B_i \subseteq B_i'$ consist of the positions where $y$ disagrees with the $0$-certificate chosen at iteration $i$. For an arbitrary nonempty $I\subseteq[k+1]$, write $j\coloneqq \min I$. Then $y^{\cup_{i\in I}B_i}$ agrees with $x$ on all positions queried before iteration $j$, and agrees with the $0$-certificate chosen at iteration $j$, and so $f(y^{\cup_{i\in I}B_i}) = 0$. Therefore, $k+1\leq \PAND(f)$. By definition, the number of queries at each iteration is at most $\C_{0,\min}^{\downarrow}(f)$. Therefore, the total number of queries is at most $\C_{0,\min}^{\downarrow}(f) \cdot \PAND(f)$, as required.
\end{proof}

Combining the lemma with ideas in the proof of \cref{lem:sdeg_lower} and the description of Step 2 in the introduction to \cref{sec:improved_lower_bound}, we obtain the following theorem.

\apxndeg*

\begin{proof}
    We can obtain a $0$-certificate for $f$ by considering $x\in f^{-1}(0)$ with minimal $\bs_x(f)$ and taking the union of the minimal sensitive blocks witnessing $\bs_x(f)$. Moreover, each minimal sensitive block of size $k$ represents a copy of the $\AND_k$ function hiding inside $f$ up to negating the input variables. (See Step 2 in the introduction to \cref{sec:improved_lower_bound} for more details.) 
    
    Now, it is easy to adapt the proof of \cref{lem:sdeg_lower} to show that $\min_{x\in f^{-1}(0)} \bs_x(f) \leq O(\ndeg_\epsilon(\neg f)^2)$. Moreover, if $f$ hides an $\AND_k$ function, then $k \leq O(\ndeg_\epsilon(f)^2)$ by a standard symmetrization argument. Therefore, 
    \begin{equation}
        \C_{0,\min}(f) \leq O(\ndeg_\epsilon(\neg f)^2\ndeg_\epsilon(f)^2),
    \end{equation}
    and the same holds with $\C_{0,\min}(f)$ replaced by $\C_{0,\min}^{\downarrow}(f)$ since the approximate nondeterministic degree cannot increase  under restriction.

    Again by a standard symmetrization argument, we have $\PAND(f) \leq O(\ndeg_\epsilon(f)^2)$. Therefore \cref{lem:pnor}, together with considering $\neg f$, gives
    \begin{equation}
        \D(f) \leq O(\ndeg_\epsilon(\neg f)^2\ndeg_\epsilon(f)^4) \quad \text{and} \quad \D(f) \leq O(\ndeg_\epsilon(f)^2\ndeg_\epsilon(\neg f)^4).
    \end{equation}
    Multiplying these inequalities and taking the square root gives the theorem.
\end{proof}

\newpage
\printbibliography

\end{document}